\documentclass[11pt]{article}
\usepackage{ifthen}
\usepackage{mdwlist}
\usepackage{amsmath,amssymb,amsfonts,amsthm}
\usepackage{bm}
\usepackage[colorlinks=true,linkcolor=blue,citecolor=blue,urlcolor=blue]{hyperref}
\usepackage{enumitem}
\usepackage{graphicx}
\usepackage{xspace}
\usepackage{verbatim}
\usepackage{algpseudocode}
\usepackage[boxruled,vlined]{algorithm2e}
	\SetKwFor{While}{While}{}{}
	\SetKwFor{For}{For}{}{}
	\SetKwIF{If}{ElseIf}{Else}{If}{}{elif}{Else}{}
	\SetKw{Return}{Return}
\usepackage[margin=1in]{geometry}
\usepackage{color}
\usepackage{thm-restate}
\usepackage{latexsym}
\usepackage{epsfig}
\def\eps{\ve}
\renewcommand{\epsilon}{\ve}
\def\ve{\varepsilon}

\newcommand{\E}{\mbox{\bf E}}

\newcommand{\pr}[2][]{\mathrm{Pr}\ifthenelse{\not\equal{}{#1}}{_{#1}}{}\!\left[#2\right]}
\newcommand{\ex}[2][]{\E\ifthenelse{\not\equal{}{#1}}{_{#1}}{}\!\left[#2\right]}

\newcommand{\dtv}{d_{\mathrm {TV}}}

\newcommand{\cW}{\mathcal{W}}
\newcommand{\cH}{\mathcal{H}}
\newcommand{\cX}{\mathcal{X}}
\newcommand{\tv}{\dtv}
\newcommand{\PC}{\text{PC}}
\newcommand{\PHS}{\text{PHS}}

\DeclareMathOperator{\argmax}{argmax}
\providecommand{\poly}{\operatorname*{poly}}

\iffalse
\newtheorem{theorem}{Theorem}
\newtheorem{nontheorem}{Non Theorem}
\newtheorem{proposition}{Proposition}
\newtheorem{observation}{Observation}
\newtheorem{remark}{Remark}
\newtheorem{fact}{Fact}
\newtheorem{lemma}{Lemma}
\newtheorem{claim}{Claim}
\newtheorem{corollary}{Corollary}
\newtheorem{case}{Case}
\newtheorem{dfn}{Definition}
\newtheorem{definition}{Definition}
\newtheorem{question}{Question}
\newtheorem{openquestion}{Open Question}
\newtheorem{res}{Result}
\newcommand{\ignore}[1]{}
\else
\newtheorem{theorem}{Theorem}

\newtheorem{proposition}[theorem]{Proposition}

\newtheorem{remark}[theorem]{Remark}

\newtheorem{lemma}[theorem]{Lemma}

\newtheorem{corollary}[theorem]{Corollary}

\newtheorem{definition}[theorem]{Definition}

\fi

\numberwithin{theorem}{section}
\numberwithin{nontheorem}{section}
\numberwithin{proposition}{section}
\numberwithin{observation}{section}
\numberwithin{remark}{section}
\numberwithin{fact}{section}
\numberwithin{lemma}{section}
\numberwithin{claim}{section}
\numberwithin{corollary}{section}
\numberwithin{case}{section}
\numberwithin{dfn}{section}
\numberwithin{definition}{section}
\numberwithin{question}{section}
\numberwithin{openquestion}{section}
\numberwithin{res}{section}

\newcommand{\ignore}[1]{}

\def\nips{0}

\if\nips0
\title{Private Hypothesis Selection\thanks{Authors are in alphabetical order. A preliminary version of this paper appeared in Advances in Neural Information Processing Systems 32 (NeurIPS 2019).}}
\else
\title{Differentially Private Hypothesis Selection}
\fi

\if\nips0
\author {
  Mark Bun\thanks{Simons Institute for the Theory of Computing and Boston University. {\tt mbun@bu.edu}. Supported by a Google Research Fellowship, as part of the Simons-Berkeley Research Fellowship program.}
  \and
  Gautam Kamath\thanks{Simons Institute for the Theory of Computing and University of Waterloo. {\tt g@csail.mit.edu}. Supported as a Microsoft Research Fellow, as part of the Simons-Berkeley Research Fellowship program. Part of this work was completed while visiting Microsoft Research, Redmond.}
  \and
  Thomas Steinke\thanks{IBM Research. \texttt{phs@thomas-steinke.net}. Part of this work completed while visiting the Simons Institute for the Theory of Computing at UC Berkeley. Now at Google.}
  \and
  Zhiwei Steven Wu\thanks{University of Minnesota, Twin Cities and Carnegie Mellon University. \texttt{zstevenwu@cmu.edu}. Part of this work completed while visiting the Simons Institute for the Theory of Computing at UC Berkeley. Supported in part by a Google Faculty Research Award, a J.P. Morgan Faculty Award, and a Facebook Research Award.}
}
\fi

\begin{document}
\maketitle

\begin{abstract}
  We provide a differentially private algorithm for hypothesis selection. 
  Given samples from an unknown probability distribution $P$ and a set of $m$ probability distributions $\mathcal{H}$, the goal is to output, in a $\varepsilon$-differentially private manner, a distribution from $\mathcal{H}$ whose total variation distance to $P$ is comparable to that of the best such distribution (which we denote by $\alpha$).
  The sample complexity of our basic algorithm is $O\left(\frac{\log m}{\alpha^2} + \frac{\log m}{\alpha \varepsilon}\right)$, representing a minimal cost for privacy when compared to the non-private algorithm. We also can handle infinite hypothesis classes $\mathcal{H}$ by relaxing to $(\varepsilon,\delta)$-differential privacy.

  We apply our hypothesis selection algorithm to give learning algorithms for a number of natural distribution classes, including Gaussians, product distributions, sums of independent random variables, piecewise polynomials, and mixture classes.
  Our hypothesis selection procedure allows us to generically convert a cover for a class to a learning algorithm, complementing known learning lower bounds which are in terms of the size of the packing number of the class.
  As the covering and packing numbers are often closely  related, for constant $\alpha$, our algorithms achieve the optimal sample complexity for many classes of interest.
  Finally, we describe an application to private distribution-free PAC learning.
\end{abstract}


\section{Introduction}
  We consider the problem of \emph{hypothesis selection}: given samples from an unknown probability distribution, select a distribution from some fixed set of candidates which is ``close'' to the unknown distribution in some appropriate distance measure.
  Such situations can arise naturally in a number of settings.
  For instance, we may have a number of different methods which work under various circumstances, which are not known in advance.
  One option is to run all the methods to generate a set of hypotheses, and pick the best from this set afterwards.
  Relatedly, an algorithm may branch its behavior based on a number of ``guesses,'' which will similarly result in a set of candidates, corresponding to the output at the end of each branch.
  Finally, if we know that the underlying distribution belongs to some (parametric) class, it is possible to essentially enumerate the class (also known as a \emph{cover}) to create a collection of hypotheses.
  Observe that this last example is quite general, and this approach can give generic learning algorithms for many settings of interest.

  This problem of hypothesis selection has been extensively studied (see, e.g.,~\cite{Yatracos85,DevroyeL96,DevroyeL97,DevroyeL01}), resulting in algorithms with a sample complexity which is \emph{logarithmic} in the number of hypotheses.
  Such a mild dependence is critical, as it facilitates sample-efficient algorithms even when the number of candidates may be large.
  These initial works have triggered a great deal of study into hypothesis selection with additional considerations, including computational efficiency, understanding the optimal approximation factor, adversarial robustness, and weakening access to the hypotheses (e.g.,~\cite{MahalanabisS08, DaskalakisDS12b, DaskalakisK14, SureshOAJ14, AcharyaJOS14b, DiakonikolasKKLMS16, AcharyaFJOS18, BousquetKM19}).

  However, in modern settings of data analysis, data may contain sensitive information about individuals.
  Some examples of such data include medical records, GPS location data, or private message transcripts.
  As such, we would like to perform statistical inference in these settings without revealing significant information about any particular individual's data.
  To this end, there have been many proposed notions of data privacy, but perhaps the gold standard is that of \emph{differential privacy}~\cite{DworkMNS06}.
  Informally, differential privacy requires that, if a single datapoint in the dataset is changed, then the distribution over outputs produced by the algorithm should be similar (see Definition~\ref{def:DP}).
  Differential privacy has seen widespread adoption, including deployment by Apple~\cite{AppleDP17}, Google~\cite{ErlingssonPK14}, and the US Census Bureau~\cite{DajaniLSKRMGDGKKLSSVA17}.

  This naturally raises the question of whether one can perform hypothesis selection under the constraint of differential privacy, while maintaining a logarithmic dependence on the size of the cover.
  Such a tool would allow us to generically obtain private learning results for a wide variety of settings.

\subsection{Results}

Our main results answer this in the affirmative: we provide differentially private algorithms for selecting a good hypothesis from a set of distributions.
The output distribution is competitive with the best distribution, and the sample complexity is bounded by the logarithm of the size of the set.
The following is a basic version of our main result.
\begin{restatable}{theorem}{tournament}
\label{thm:tournament}
  Let $\mathcal{H} = \{H_1, \dots, H_m\}$ be a set of probability
  distributions.  Let $D = \{X_1, \dots, X_n\}$ be a set of samples
  drawn independently from an unknown probability distribution $P$.
  There exists an $\ve$-differentially private
  algorithm (with respect to the dataset $D$) which has following
  guarantees.  Suppose there exists a distribution
  $H^* \in \mathcal{H}$ such that $\tv(P, H^*) \leq \alpha$.  If
  $n = \Omega\left(\frac{\log m}{\alpha^2} + \frac{\log m}{\alpha
      \ve}\right)$, then the algorithm will output a distribution
  $\hat H \in \mathcal{H}$ such that $\dtv(P, \hat H) \leq (3+\zeta)\alpha$ with probability at least $9/10$, for any constant $\zeta > 0$.
  The running time of the algorithm is $O(nm^2)$.
\end{restatable}
The sample complexity of this problem without privacy constraints is $O\left(\frac{\log m}{\alpha^2}\right)$, and thus the additional cost for $\varepsilon$-differential privacy is an additive $O\left(\frac{\log m}{\alpha \ve}\right)$.
We consider this cost to be minimal; in particular, the dependence on $m$ is unchanged.
Note that the running time of our algorithm is $O(nm^2)$ -- we conjecture it may be possible to reduce this to $\tilde O(nm)$ as has been done in the non-private setting~\cite{DaskalakisK14,SureshOAJ14,AcharyaJOS14b,AcharyaFJOS18}, though we have not attempted to perform this optimization.
Regardless, our main focus is on the sample complexity rather than the running time, since any method for generic hypothesis selection requires $\Omega(m)$ time, thus precluding efficient algorithms when $m$ is large.
Note that the approximation factor of $(3+\zeta)\alpha$ is effectively tight. 
That is, even in the infinite sample limit and without the constraint of privacy, information theoretically, one can not achieve a better approximation than $3\alpha$~\cite{DevroyeL01, MahalanabisS08}.\footnote{Note that this can be brought down to $2\alpha$ if one instead outputs a \emph{mixture} of $H_i \in \mathcal{H}$~\cite{ BousquetKM19}.}
Theorem~\ref{thm:tournament} requires prior knowledge of the value of $\alpha$, though we can use this to obtain an algorithm with similar guarantees which does not (Theorem~\ref{thm:what-did-i-do-to-deserve-this}).

It is possible to improve the guarantees of this algorithm in two ways (Theorem~\ref{thm:sel-inf}).
First, if the distributions are nicely structured, the former term in the sample complexity can be reduced from $O(\log m/\alpha^2)$ to $O(d/\alpha^2)$, where $d$ is a VC-dimension-based measure of the complexity of the collection of distributions.
Second, if there are few hypotheses which are close to the true distribution, then we can pay only logarithmically in this number, as opposed to the total number of hypotheses.
These modifications allow us to handle instances where $m$ may be very large (or even infinite), albeit at the cost of weakening to approximate differential privacy to perform the second refinement.
A technical discussion of our methods is in Section~\ref{sec:techniques}, our basic approach is covered in Section~\ref{sec:hyp-sel}, and the version with all the bells and whistles appears in Section~\ref{sec:infinite}.

From Theorem~\ref{thm:tournament}, we immediately obtain Corollary~\ref{cor:covertolearn} which applies when $\mathcal{H}$ itself may not be finite, but admits a finite cover with respect to total variation distance.
\begin{restatable}{corollary}{covertolearn}
  \label{cor:covertolearn}
  Suppose there exists an $\alpha$-cover $\mathcal{C}_\alpha$ of a set of distributions $\mathcal{H}$, and that we are given a set of samples $X_1, \dots, X_n \sim P$, where $\dtv(P, \mathcal{H}) \leq \alpha$.
  For any constant $\zeta > 0$, there exists an $\ve$-differentially private algorithm (with respect to the input $\{X_1, \dots, X_n\}$) which outputs a distribution $H^* \in \mathcal{C}_\alpha$ such that $\dtv(P, H^*) \leq (6 + 2\zeta)\alpha$ with probability $\geq 9/10$, as long as
  $$n = \Omega\left(\frac{\log |\mathcal{C}_\alpha|}{\alpha^2} + \frac{\log |\mathcal{C}_\alpha|}{\alpha \varepsilon }\right).$$
\end{restatable}

Informally, this says that if a hypothesis class has an $\alpha$-cover $\mathcal{C}_\alpha$, then there is a private learning algorithm for the class which requires $O(\log |\mathcal{C}_\alpha|)$ samples.
Note that our algorithm works even if the unknown distribution is only \emph{close} to the hypothesis class.
This is useful when we may have model misspecification, or when we require adversarial robustness. (We also give an extension of this algorithm which gives guarantees in the \emph{semi-agnostic} learning model; see Section~\ref{sec:semi} for details.)
The requirements for this theorem to apply are minimal, and thus it generically provides learning algorithms for a wide variety of hypothesis classes.
That said, in non-private settings, the sample complexity given by this method is rather lossy: as an extreme example, there is no finite-size cover of univariate Gaussian distributions with unbounded parameters, so this approach does not give a finite-sample algorithm.
That said, it is well-known that $O(1/\alpha^2)$ samples suffice to estimate a Gaussian in total variation distance.
In the private setting, our theorem incurs a cost which is somewhat necessary: in particular, it is folklore that any pure $\ve$-differentially private learning algorithm must pay a cost which is logarithmic in the packing number of the class (for completeness, see Lemma~\ref{lem:packing}).
Due to the relationship between packing and covering numbers (Lemma~\ref{lem:pack-n-cover}), this implies that up to a constant factor relaxation in the learning accuracy, our results are tight (Theorem~\ref{thm:packing-tight}).
Further discussion appears in Sections~\ref{sec:packing}.

Given Corollary~\ref{cor:covertolearn}, in Section~\ref{sec:apps}, we derive new learning results for a number of classes.
Our main applications are for $d$-dimensional Gaussian and product distributions.
Informally, we obtain $\tilde O(d)$ sample algorithms for learning a product distribution and a Gaussian with known covariance (Corollaries~\ref{cor:learn-prod} and~\ref{cor:learn-gauss-mean}), and an $\tilde O(d^2)$ algorithm for learning a Gaussian with unknown covariance (Corollary~\ref{cor:learn-gauss-cov}).
These improve on recent results by Kamath, Li, Singhal, and Ullman~\cite{KamathLSU19} in two different ways.
First, as mentioned before, our results are semi-agnostic, so we can handle when the distribution is only \emph{close} to a product or Gaussian distribution.
Second, our results hold for pure $(\ve,0)$-differential privacy, which is a stronger notion than $\ve^2$-zCDP as considered in~\cite{KamathLSU19}.
In this weaker model, they also obtained $\tilde O(d)$ and $\tilde O(d^2)$ sample algorithms, but the natural modifications to achieve $\ve$-DP incur extra $\poly(d)$ factors.\footnote{Roughly, this is due to the fact that the Laplace and Gaussian mechanism are based on $\ell_1$ and $\ell_2$ sensitivity, respectively, and that there is a $\sqrt{d}$-factor relationship between these two norms, in the worst case.}
\cite{KamathLSU19} also showed $\tilde \Omega(d)$ lower bounds for Gaussian and product distribution estimation in the even weaker model of $(\varepsilon, \delta)$-differential privacy. Thus, our results show that the dimension dependence for these problems is unchanged for essentially any notion of differential privacy. 
In particular, our results show a previously-unknown separation between mean estimation of product distributions and non-product distributions under pure $(\varepsilon,0)$-differential privacy; see Remark~\ref{rem:separation}.

We also apply Theorem~\ref{thm:sel-inf} to obtain algorithms for learning Gaussians under $(\ve,\delta)$-differential privacy, with no bounds on the mean and variance parameters.
More specifically, we provide algorithms for learning multivariate Gaussians with unknown mean and known covariance (Corollary~\ref{cor:learn-gauss-mean2}), and univariate Gaussians with both unknown mean and variance (Corollary~\ref{cor:learn-gaussian-univ-infinite}).
For the former problem, we manage to avoid dependences which arise due to the application of advanced composition (similar to Remark~\ref{rem:separation}).

To demonstrate the flexibility of our approach, we also give private learning algorithms for sums of independent random variables (Corollaries~\ref{cor:learn-siirv} and~\ref{cor:learn-pmd}) and piecewise polynomials (Corollary~\ref{cor:learn-pp}).
To the best of our knowledge, the former class of distributions has not been considered in the private setting, and we rely on covering theorems from the non-private literature.
Private learning algorithms for the latter class, piecewise polynomials, have been studied by Diakonikolas, Hardt, and Schmidt~\cite{DiakonikolasHS15}.
They provide sample and time efficient algorithms for histogram distributions (i.e., piecewise constant distributions), and claim similar results for general piecewise polynomials.
Their method depends heavily on rather sophisticated algorithms for the non-private version of this problem~\cite{AcharyaDLS17}.
In constrast, we can obtain comparable sample complexity bounds from just the existence of a cover and elementary VC dimension arguments, which we derive in a fairly self-contained manner.

We additionally give algorithms for learning mixtures of any coverable class (Corollary~\ref{cor:learn-mixture}).
In particular, this immediately implies algorithms for learning mixtures of Gaussians, product distributions, and all other classes mentioned above.

To conclude our applications, we discuss a connection to PAC learning (Corollary~\ref{cor:pac}). It is known that the sample complexity of differentially private distribution-free PAC learning can be higher than that of non-private learning. However, this gap does not exist for distribution-specific learning, where the learning algorithm knows the distribution of (unlabeled) examples, as both sample complexities are characterized by VC dimension. Private hypothesis selection allows us to address an intermediate situation where the distribution of unlabeled examples is not known exactly, but is known to come (approximately) from a class of distributions. When this class has a small cover, we are able to recover sample complexity guarantees for private PAC learning which are comparable to the non-private case.

\subsection{Techniques}
\label{sec:techniques}
Non-privately, most algorithms for hypothesis selection involve a tournament-style approach.
We conduct a number of pairwise comparisons between distributions, which may either have a winner and a loser, or may be declared a draw.
Intuitively, a distribution will be declared the winner of a comparison if it is much closer than the alternative to the unknown distribution, and a tie will be declared if the two distributions are comparably close.
The algorithm will output any distribution which never loses a comparison.
A single comparison between a pair of hypotheses requires $O(1/\alpha^2)$ samples, and a Chernoff plus union bound argument over the $O(m^2)$ possible comparisons increases the sample complexity to $O(\log m/\alpha^2)$.
In fact, we can use uniform convergence arguments to reduce this sample complexity to $O(d/\alpha^2)$, where $d$ is the VC dimension of the $2 \binom{m}{2}$ sets (the ``Scheff\'e'' sets) defined by the subsets of the domain where the PDF of one distribution dominates another.
Crucially, we must reuse the same set of samples for all comparisons to avoid paying polynomially in the number of hypotheses.

A private algorithm for this problem requires additional care.
Since a single comparison is based on the number of samples which fall into a particular subset of the domain, the sensitivity of the underlying statistic is low, and thus privacy may seem easily achievable at first glance.
However, the challenge comes from the fact that the same samples are reused for all pairwise comparisons, thus greatly increasing the sensitivity: changing a single datapoint could flip the result of every comparison!
In order to avoid this pitfall, we instead carefully construct a score function for each hypothesis, namely, the minimum number of points that must be changed to cause the distribution to lose any comparison.
For this to be a useful score function, we must show that the best hypothesis will win all of its comparisons by a large margin.
We can then use the Exponential Mechanism~\cite{McSherryT07} to select a distribution with high score.

Further improvements can be made if we are guaranteed that the number of ``good'' hypotheses (i.e., those that have total variation distance from the true distribution bounded by $(3+\zeta)\alpha$) is at most some parameter $k$, and if we are willing to relax to approximate differential privacy. The parameter $k$ here is related to the doubling dimension of the hypothesis class with respect to total variation distance.
If we randomly assign the hypotheses to $\Omega(k^2)$ buckets, with high probability, no bucket will contain more than one good hypothesis.
We can identify a bucket containing a good hypothesis using a similar method based on the exponential mechanism as described above.
Moreover, since we are likely to only have one ``good'' hypothesis in the chosen bucket, this implies a significant gap between the best and second-best scores in that bucket. This allows us to use stability-based techniques~\cite{DworkL09, ThakurtaS13}, and in particular the GAP-MAX algorithm of Bun, Dwork, Rothblum, and Steinke~\cite{BunDRS18}, to identify an accurate distribution.

\subsection{Related Work}
Our main result builds on a long line of work on non-private hypothesis selection.
One starting point for the particular style of approach we consider here is~\cite{Yatracos85}, which was expanded on in~\cite{DevroyeL96,DevroyeL97,DevroyeL01}.
Since then, there has been study into hypothesis selection under additional considerations, including computational efficiency, understanding the optimal approximation factor, adversarial robustness, and weakening access to the hypotheses~\cite{MahalanabisS08, DaskalakisDS12b, DaskalakisK14, SureshOAJ14, AcharyaJOS14b, DiakonikolasKKLMS16, AcharyaFJOS18, BousquetKM19}.
Our private algorithm examines the same type of problem, with the additional constraint of differential privacy.

Perhaps the most closely related work is that of Canonne, Kamath, McMillan, Smith, and Ullman~\cite{CanonneKMSU19}, which focuses on the case of private simple hypothesis testing.
This is a more restricted setting than we consider in this paper, as it focuses on the case where we are trying to decide between $m = 2$ hypotheses, and we are guaranteed that the unknown distribution is one of these two hypotheses.
However, in this setting, they are able to get an instance-by-instance characterization of the sample complexity, depending on both the total variation and Hellinger distance between the two distributions.

There has recently been a great deal of interest in differentially private distribution learning.
In the central model, most relevant are~\cite{DiakonikolasHS15}, which gives algorithms for learning structured univariate distributions, and~\cite{KarwaV18,KamathLSU19}, which focus on learning Gaussians and binary product distributions.
\cite{CaiWZ19} also studies private statistical parameter estimation.
Privately learning mixtures of Gaussians was considered in~\cite{NissimRS07,KamathSSU19}.
The latter paper (which is concurrent with the present work) gives a computationally efficient algorithm for the problem, but with a worse sample complexity, and incomparable accuracy guarantees (they require a separation condition, and perform clustering and parameter estimation, while we do proper learning).
\cite{BunNSV15} give an algorithm for learning distributions in Kolmogorov distance.
Upper and lower bounds for learning the mean of a product distribution over the hypercube in $\ell_\infty$-distance include~\cite{BlumDMN05, BunUV14, DworkMNS06, SteinkeU17a}.
\cite{AcharyaKSZ18} focuses on estimating properties of a distribution, rather than the distribution itself.
\cite{Smith11} gives an algorithm which allows one to estimate asymptotically normal statistics with optimal convergence rates, but no finite sample complexity guarantees.
There has also been a great deal of work on distribution learning in the local model of differential privacy~\cite{DuchiJW13,WangHWNXYLQ16,KairouzBR16,AcharyaSZ19, DuchiR18,JosephKMW19,YeB18,GaboardiRS19}.
For further coverage of differentially private statistics, see~\cite{KamathU20}.

Non-privately, there has been a significant amount of work on learning specific classes of distributions.
The PAC-style formulation of the problem we consider originated in~\cite{KearnsMRRSS94}.
While learning Gaussians and product distributions can be considered folklore at this point, some of the other classes we learn have enjoyed more recent study.
For instance, learning sums of independent random variables was recently considered in~\cite{DaskalakisDS12b} toward the problem of learning Poisson Binomial Distributions (PBDs).
Since then, there has been additional work on learning PBDs and various generalizations~\cite{DaskalakisKT15, DaskalakisDKT16, DiakonikolasKS16a, DiakonikolasKS16b, DiakonikolasKS16c,DeLS18}.

Piecewise polynomials are a highly-expressive class of distributions, and they can be used to approximate a number of other univariate distribution classes, including distributions which are multi-modal, concave, convex, log-concave, monotone hazard rate, Gaussian, Poisson, Binomial, and more.
Algorithms for learning such classes are considered in a number of papers, including~\cite{DaskalakisDS12a,ChanDSS14a,ChanDSS14b,AcharyaDK15, AcharyaDLS17}.

There has also been a great deal of work on learning mixtures of distribution classes, particularly mixtures of Gaussians.
There are many ways the objective of such a problem can be defined, including clustering~\cite{Dasgupta99, DasguptaS00, AroraK01, VempalaW02, AchlioptasM05, ChaudhuriR08a, ChaudhuriR08b, KumarK10, AwasthiS12, RegevV17, HopkinsL18, DiakonikolasKS18b, KothariSS18}, parameter estimation~\cite{KalaiMV10, MoitraV10, BelkinS10,HsuK13,AndersonBGRV14,BhaskaraCMV14, HardtP15,GeHK15, XuHM16, DaskalakisTZ17,AshtianiBHLMP18}, proper learning~\cite{FeldmanOS06,FeldmanOS08, DaskalakisK14, SureshOAJ14, DiakonikolasKKLMS16, LiS17}, and improper learning~\cite{ChanDSS14a}.
Our work falls into the line on proper learning: the algorithm is given a set of samples from a mixture of Gaussians, and must output a mixture of Gaussians which is close in total variation distance.

\subsubsection{Subsequent Work}
Since the initial appearance of this work, there have been several relevant results in the surrounding area.
Most pertinent is the work of Aden-Ali, Ashtiani, and Kamath~\cite{AdenAliAK21}, which improves our main hypothesis selection result by improving the constant approximation factor, making the algorithm agnostic, and having a simpler analysis.
The authors also argue that, given a small cover for each distribution in the space, there exists a cover which is locally-small everywhere as required by Theorem~\ref{thm:sel-inf}, thus allowing them to learn unbounded Gaussians with arbitrary covariances.

Our results and techniques have seen use for other problems as well.
Kamath, Singhal, and Ullman~\cite{KamathSU20} use a similar approach based on pairwise comparisons in order to perform mean estimation, rather than our work which focuses on density estimation.
Liu et al.~\cite{LiuSYKR20} use our main algorithm to estimate discrete distributions in a type of federated learning setting.

Hypothesis selection has since been studied in other privacy models.
Informally speaking, our work shows that the sample complexity of hypothesis selection under central differential privacy is $O(\log m)$.
Gopi et al.~\cite{GopiKKNWZ20} study the problem under the stronger notion of local differential privacy~\cite{Warner65,EvfimievskiGS03,KasiviswanathanLNRS11}, showing that the sample complexity is $\tilde \Theta(m)$, an exponential increase in cost.
Under various notions of pan-privacy~\cite{DworkNPRY10} and the shuffled model~\cite{CheuSUZZ19, ErlingssonFMRTT19}, which are intermediate to local and central differential privacy, the sample complexity of $m$-wise simple hypothesis testing (the easier version of hypothesis selection, where the unknown distribution is equal to one of the given distributions) was shown to be $\tilde \Theta(\sqrt{m})$~\cite{CheuU20}.

\subsection{Organization}
We begin in Section~\ref{sec:preliminaries} with preliminaries.
In Section~\ref{sec:hyp-sel}, we give a basic algorithm for private hypothesis selection, via the exponential mechanism.
In Section~\ref{sec:infinite}, we extend this approach in two ways: by using VC dimension arguments to reduce the sample complexity for sets of hypotheses with additional structure, and combining this with a GAP-MAX algorithm to achieve non-trivial guarantees for infinite hypothesis classes.
Section~\ref{sec:packing} shows that our approach leads to algorithms which essentially match lower bounds for most distribution classes (in the constant $\alpha$ regime).
We consider applications in Section~\ref{sec:apps}: through a combination of arguments about covers and VC dimension, we derive algorithms for learning a number of classes of distributions, as well as describe an application to private PAC learning.
Finally, we conclude in Section~\ref{sec:conclusions} with open questions.


\section{Preliminaries}
\label{sec:preliminaries}
We start with some preliminaries and definitions.

\begin{definition}
  The \emph{total variation distance} or \emph{statistical distance} between $P$ and $Q$ is defined as
  $$\dtv(P,Q) = \max_{S \subseteq \Omega} P(S) - Q(S) = \frac{1}{2} \int_{x \in \Omega} |P(x) - Q(x)| dx = \frac{1}{2}\|P - Q\|_1 \in [0,1]. $$
  Moreover, if $\mathcal{H}$ is a set of distributions over a common domain, we define $\dtv(P, \mathcal{H}) = \inf_{H\in \mathcal{H}} \dtv(P, H)$.
\end{definition}

Throughout this paper, we consider packings and coverings of sets of distributions with respect to total variation distance.
\begin{definition}
  A \emph{$\gamma$-cover} of a set of distributions $\mathcal{H}$ is a set of distributions $\mathcal{C}_\gamma$, such that for every $H \in \mathcal{H}$, there exists some $P \in \mathcal{C}_\gamma$ such that $\dtv(P,H) \leq \gamma$.
  
  A \emph{$\gamma$-packing} of a set of distributions $\mathcal{H}$ is a set of distributions $\mathcal{P}_\gamma \subseteq \mathcal{H}$, such that for every pair of distributions $P, Q \in \mathcal{P}_\gamma$, we have that $\dtv(P,Q) > \gamma$.
\end{definition}

In this paper, we present semi-agnostic learning algorithms.
\begin{definition}
  An algorithm is said to be an $\alpha$-\emph{semi-agnostic learner} for a class $\mathcal{H}$ if it has the following guarantees.
  Suppose we are given $X_1, \dots, X_n \sim P$, where $\dtv(P,\mathcal{H}) \leq \mathrm{OPT}$.
  The algorithm must output some distribution $\hat H$ such that $\dtv(P,H) \leq c \cdot \mathrm{OPT} + O(\alpha)$, for some constant $c \geq 1$.
  If $c = 1$, then the algorithm is said to be agnostic.
\end{definition}

Now we define differential privacy. We say that $D$ and $D'$ are
neighboring datasets, denoted $D \sim D'$, if $D$ and $D'$ differ by
at most one observation. Informally, differential privacy requires
that the algorithm has close output distributions when run
on any pair of neighboring datasets. More formally:

\begin{definition}[\cite{DworkMNS06}]
  \label{def:DP}
  A randomized algorithm $T : X^* \rightarrow \mathcal{R}$ is
  \emph{$(\eps, \delta)$-differentially private} if for all $n\geq 1$,
  for all neighboring datasets $D,D'\in X^n$, and for all events
  $S\subseteq \mathcal{R}$,
  $$\Pr\left[T(D)\in S\right]\leq e^\eps \Pr[T(D')\in S] + \delta\,. $$
  If $\delta = 0$, we say that $T$ is $\eps$-differentially private.
\end{definition}

We will also use the related notion of concentrated differential privacy:
\begin{definition}[\cite{DworkR16,BunS16}]
  A randomized algorithm $T : X^* \rightarrow \mathcal{R}$ satisfies
  \emph{$\rho$-zero-concentrated differential privacy} if for all $n\geq 1$,
  for all neighboring datasets $D,D'\in X^n$, and for all $\alpha \in (1, \infty)$,
  $$R_\alpha(M(D)||M(D')) \leq \rho\alpha,$$
  where $R_\alpha(M(D)||M(D'))$ is the $\alpha$-R\'enyi divergence between $M(D)$ and $M(D')$.\footnote{Given two probability distributions $P, Q$ over $\Omega$, $R_\alpha(P||Q) = \frac{1}{\alpha -1}\log\left(\sum_{x \in \Omega} P(x)^\alpha Q(x)^{1 - \alpha}\right)$.}
\end{definition}

The exponential mechanism~\cite{McSherryT07} is a powerful $\ve$-differentially private
mechanism for selecting an approximately best outcome from a set of
alternatives, where the quality of an outcome is measured by a score
function relating each alternative to the underlying dataset. Letting
$\mathcal{R}$ be the set of possible outcomes, a score function
$q : X^* \times \mathcal{R} \rightarrow \mathbb{R}$ maps each pair consisting of a dataset
and an outcome to a real-valued score.  The exponential mechanism
$\mathcal{M}_E$ instantiated with a dataset $D$, a score function
$q$, and a privacy parameter $\eps$ selects an outcome $r$ in
$\mathcal{R}$ with probability proportional to
$\exp\left(\eps q(D, r) / (2\Delta(q))\right)$, where $\Delta(q)$ is
the sensitivity of the score function defined as
\[
  \Delta(q) = \max_{r\in \mathcal{R}, D\sim D'} \left|q(D, r) - q(D',
    r) \right|.
\]

\begin{theorem}[\cite{McSherryT07}]\label{exp}
  For any input dataset $D$, score function $q$ and privacy parameter
  $\eps >0$, the exponential mechanism $\mathcal{M}_E(D, q, \eps)$ is
  $\eps$-differentially private, and with probability at least
  $1 - \beta$, selects an outcome $r\in \mathcal{R}$ such that
  \[
    q(D, r) \geq \max_{r'\in \mathcal{R}} q(D, r') - \frac{2\Delta(q)
      \log(|\mathcal{R}|/\beta)}{\eps}.
  \]

\end{theorem}



\section{A First Method for Private Hypothesis Selection}
\label{sec:hyp-sel}

In this section, we present our first algorithm for private hypothesis
selection and obtain the following result.

\tournament*

Note that the sample complexity bound above scales logarithmically
with the size of the hypothesis class. 
In Section~\ref{sec:infinite},
we will provide a stronger result (which subsumes the present one as a special case) that can handle certain infinite
hypothesis classes.
For sake of exposition, we begin in this section with the basic algorithm.

\subsection{Pairwise Comparisons}

We first present a subroutine which compares two hypothesis
distributions. 
This subroutine is due to Daskalakis, Diakonikolas, and Servedio~\cite{DaskalakisDS12b}, and is essentially a modification of previous methods (e.g.,~\cite{DevroyeL01}) to allow for draws.
Let $H$ and $H'$ be two distributions over domain $\cX$ and consider
the following set, which is called the \emph{Scheff\'e set}:
\[
  \cW_1 = \{x\in \cX \mid H(x) > H'(x)\}
\]
Define $p_1 = H(\cW_1)$, $p_2 = H'(\cW_1)$, and $\tau = P(\cW_1)$ to be the probability masses that $H$, $H'$, and $P$ place on $\cW_1$, respectively. It follows that
$p_1 > p_2$ and $p_1 - p_2 = \tv(H, H')$.\footnote{For simplicity
  of our exposition, we will assume that we can evaluate the two
  quantities $p_1$ and $p_2$ exactly. In general, we can estimate
  these quantities to arbitrary accuracy, as long as, for each hypothesis $H$, we can evaluate
  the density of each point under $H$ and also draw samples from
  $H$.}

\begin{algorithm}[H]
  \SetAlgoLined
  \textbf{Input}: Two hypotheses $H$ and $H'$, input dataset $D$ of size $n$ drawn i.i.d.\ from target distribution $P$, approximation parameter $\zeta > 0$, and accuracy parameter $\alpha \in (0, 1)$.\\
  \textbf{Initialize}: Compute the fraction of points that fall into
  $\cW_1$:
  $\hat \tau = \frac{1}{n} \left|\{ x\in D \mid x \in \cW_1 \}\right|$.
  \\
  \textbf{If} $p_1 - p_2 \leq (2+\zeta) \alpha$, return ``Draw''.\\
  \textbf{Else If} $\hat \tau > p_1 - (1+\zeta/2)\alpha$, return $H$ as the winner.\\
  \textbf{Else If} $\hat \tau < p_2 + (1+\zeta/2)\alpha$, return $H'$ as the winner. \\
  \textbf{Else} return ``Draw''.
  \caption{\textsc{Pairwise Contest}: $\PC(H, H', D, \zeta, \alpha)$}
  \label{alg:friday}
\end{algorithm}

\noindent Now consider the following function of this ordered pair of
hypotheses:
\begin{equation*}
  \Gamma_{\zeta}(H, H', D) = \begin{cases}
    n \qquad &\text{if } p_1 - p_2 \leq (2+\zeta)\alpha;\\
    n \cdot \, \max\{0, \hat\tau - (p_2 + (1 + \zeta/2)\alpha)\}  \qquad&\text{otherwise.}
\end{cases}
\end{equation*}

When the two hypotheses are sufficiently far apart (i.e., 
$\tv(H, H') > (2+\zeta)\alpha$), $\Gamma_{\zeta}(H, H', D)$ is essentially the
number of points one needs to change in $D$ to make $H'$ the
winner.

\begin{lemma}\label{friday}
  Let $P,H,H'$ be distributions as above. With probability at least
  $1 - 2\exp(- n \zeta^2\alpha^2/8)$ over the random draws of $D$ from
  $P^n$, $\hat \tau$ satisfies $|\hat \tau - \tau| < \zeta\alpha/4$,
  and if $\tv(P, H) \leq \alpha$, then
  $\Gamma_{\zeta}(H, H', D) > \zeta\alpha n/4$.
\end{lemma}

\begin{proof}
  By applying Hoeffding's inequality, we know that with
  probability at least $1 - 2\exp(- n\zeta^2\alpha^2/8)$,
  $|\tau - \hat \tau| < \zeta\alpha/4$. We condition on this event for
  the remainder of the proof. Consider the following two
  cases. In the first case, suppose that $p_1 - p_2 \leq (2+\zeta)\alpha$.
  Then we know that $\Gamma_{\zeta}(H, H', D) = n > \alpha n$. In the
  second case, suppose that $p_1 - p_2 > (2+\zeta)\alpha$. Since
  $\tv(P, H) \leq \alpha$, we know that $|p_1 - \tau| \leq \alpha$,
  and so $|p_1 - \hat \tau| < (1+\zeta/4)\alpha$. Since $p_1 > p_2 + (2+\zeta)\alpha$,
  we also have $\hat \tau > p_2 + (1 + 3\zeta/4)\alpha$. It follows that
  $\Gamma_{\zeta}(H, H', D) = n (\hat \tau - (p_2+ (1+\zeta/2)\alpha)) > \zeta\alpha
  n/4$.
\end{proof}

\subsection{Na\"ive Approach via Laplace Mechanism}
We first sketch a na\"ive approach for private hypothesis selection, based on the primitive in Algorithm~\ref{alg:friday}.
This is a privatization of a similar approach which appeared in~\cite{DaskalakisK14}, though the idea behind the approach is older, e.g.,~\cite{DevroyeL01} -- the~\cite{DaskalakisK14} approach differs slightly since it employs a comparison procedure which allows ties, as we do.
Later, Lemma~\ref{lem:terrible} describes the approach and privatization of~\cite{DevroyeL01} in more detail, which are morally equivalent to what we discuss here.

A non-private algorithm for selection from $m$ hypotheses would run Algorithm~\ref{alg:friday} on each pair of hypotheses, either outputting a winner between the two distributions, or declaring a tie in the case when the total variation distance between the two distributions is small.
The algorithm would output any distribution which never loses a comparison.
Correctness of this algorithm relies on the empirical masses in all $O(m^2)$ Scheff\'e sets being estimated up to an additive $O(\alpha)$, which, by Hoeffding's inequality, happens with constant probability when $n \geq O\left(\frac{\log m}{\alpha^2}\right)$.
Crucially, we reuse the same set of samples for all comparisons.
With this in hand, it is not hard to show that a distribution $H$ which is $\alpha$-close to $P$ will never lose a comparison, and any distribution $H'$ which is $c\alpha$-far from $P$ (for an appropriately chosen constant $c > 1$) will lose its comparison with $H$, thus ensuring that the winning distribution will be $c\alpha$-close to $P$.

Now, we consider how to privatize this algorithm.
Each of the $O(m^2)$ comparisons is based on the quantity $\hat \tau = \frac{1}{n} \left|\{ x\in D \mid x \in \cW \}\right|$, where $\cW$ is the Scheff\'e set between the two distributions $H$ and $H'$.
To make a single comparison $\ve$-differentially private, we would have to add Laplace noise of order $O\left(\frac{1}{\ve n}\right)$ to this quantity.
However, since we reuse the same set of samples for all comparisons, in order to make the result of all $O(m^2)$ comparisons $\ve$-differentially private, the basic composition property of differential privacy would prescribe adding Laplace noise of order $O\left(\frac{m^2}{\ve n}\right)$ to the quantity used in each comparison.
To bound the noise error of all comparisons simultaneously by $O(\alpha)$, we thus require $n \geq O\left(\frac{m^2 \log m}{\alpha\ve}\right)$, and the rest of the analysis is then identical to before.

A formalization of this argument allows us to arrive at the following theorem.
The accuracy bound is of the appropriate form, but the cost of privacy is an exponential increase in the sample complexity.

\begin{theorem}
\label{thm:bad-tournament}
  Let $\mathcal{H} = \{H_1, \dots, H_m\}$ be a set of probability
  distributions.  Let $D = \{X_1, \dots, X_n\}$ be a set of samples
  drawn independently from an unknown probability distribution $P$.
  There exists an $\ve$-differentially private
  algorithm (with respect to the dataset $D$) which has following
  guarantees.  Suppose there exists a distribution
  $H^* \in \mathcal{H}$ such that $\tv(P, H^*) \leq \alpha$.  If
  $n = \Omega\left(\frac{\log m}{\alpha^2} + \frac{m^2 \log m}{\alpha
      \ve}\right)$, then the algorithm will output a distribution
  $\hat H \in \mathcal{H}$ such that $\dtv(P, \hat H) \leq O(\alpha)$ with probability at least $9/10$.
  The running time of the algorithm is $O(nm^2)$.
\end{theorem}

\subsection{Selection via Exponential Mechanism}
\label{sec:exp-mech}
In light of the definition of the pairwise comparison defined above, we
consider the following score function
$S \colon \mathcal{H} \times \mathcal{X}^n$, such that for any
$H_j\in \cH$ and dataset $D$,
\begin{equation}
  S(H_j, D) = \min_{H_k\in \cH} \Gamma_{\zeta}(H_j, H_k, D). \label{score}
\end{equation}
Roughly speaking, $S(H_j, D)$ is the minimum number of points required
to change in $D$ in order for $H_j$ to lose at least one pairwise
contest against a different hypothesis. When the hypothesis $H_j$
is very close to every other distribution, such that all pairwise
contests return ``Draw,'' then the score will be $n$.

\begin{algorithm}[H] \label{alg:phs}
  \SetAlgoLined \textbf{Input}: Dataset $D$, a collection of
  hypotheses
  $\cH =\{H_1, \ldots, H_m\}$, privacy parameter $\epsilon$.\\
  Output a random hypothesis $\hat H \in \cH$ such that for each $H_j$
  \[
    \Pr[\hat H = H_j] \propto \exp\left(\frac{S(H_j, D)}{2\epsilon}
    \right)
  \]
  where $S(H_j, D)$ is defined in \eqref{score}.
  \caption{\textsc{Private Hypothesis Selection}:
    $\PHS(\cH, D, \epsilon)$}
\end{algorithm}

\begin{lemma}[Privacy]\label{lem:privacy}
  For any $\epsilon > 0$ and collection of hypotheses $\cH$, the
  algorithm $\PHS(\cH, \cdot, \epsilon)$ satisfies
  $\epsilon$-differential privacy.
\end{lemma}

\begin{proof}
  First, observe that for any pairs of hypotheses $H_j, H_k$,
  $\Gamma_{\zeta}(H_j, H_k, \cdot)$ has sensitivity 1. As a result, the score
  function $S$ is also 1-sensitive. Then the result directly follows
  from the privacy guarantee of the exponential mechanism (Theorem \ref{exp}).
\end{proof}

\begin{lemma}[Utility] \label{lem:tournament-utility} Fix any
  $\alpha, \beta \in (0, 1)$, and $\zeta > 0$. Suppose that there
  exists $H^*\in \cH$ such that $\tv(P, H^*) \leq \alpha$. Then with
  probability $1 - \beta$ over the sample $D$ and the algorithm
  $\PHS$, we have that $\PHS(\cH, D)$ outputs an hypothesis $\hat H$
  such that $\tv(P, \hat H) \leq (3+\zeta)\alpha$, as long as the
  sample size satisfies
  \[
    n\geq \frac{8 \ln(4m/\beta)}{\zeta^2\alpha^2} +
    \frac{8\ln(2m/\beta)}{\zeta \alpha \epsilon}.
  \]
\end{lemma}

\begin{proof}
  First, consider the $m$ pairwise contests between $H^*$ and every
  candidate in $\cH$. Let $\cW_j = \{x\in \cX\mid H^*(x) > H_j(x)  \}$
  be the collection of Scheff\'{e} sets. For any event $W\subseteq \cX$,
  let $\hat P(W)$ denote the empirical probability of event $W$ on the
  dataset $D$. By Lemma~\ref{friday} and an application of the union
  bound, we know that with probability at least $1 - 2m \exp(-n\zeta^2\alpha^2/8)$
  over the draws of $D$, $|P(\cW_j) - \hat P(\cW_j)| \leq \zeta\alpha/4$ 
  and $\Gamma_{\zeta}(H^*, H_j, D) > \zeta \alpha n/4$ for all $H_j\in \cH$. In
  particular, the latter event implies that $S(H^*, D) > \zeta\alpha n/4$.

  Next, by the utility guarantee of the exponential mechanism
  (Theorem~\ref{exp}), we know that with probability at least
  $1 - \beta/2$, the output hypothesis satisfies
  \[
    S(\hat H, D) \geq S(H^*, D) - \frac{2 \ln(2m/\beta)}{\epsilon} >
    \zeta\alpha n/4 - \frac{2 \ln(2m/\beta)}{\epsilon}.
  \]
  Then as long as
  $n \geq \frac{8 \ln(4m/\beta)}{\zeta^2 \alpha^2} +
  \frac{8\ln(2m/\beta)}{\zeta \alpha \epsilon}$, we know that with
  probability at least $1 - \beta$, $S(\hat H, D) > 0$.  Let us
  condition on this event, which implies that
  $\Gamma_{\zeta}(\hat H, H^*, D) > 0$. We will now show that
  $\tv(\hat H, H^*) \leq (2+\zeta)\alpha$, which directly implies that
  $\tv(\hat H, P) \leq (3 + \zeta)\alpha$ by the triangle inequality. Suppose to the
  contrary that $\tv(\hat H, H^*) > (2+\zeta)\alpha$. Then by the definition of
  $\Gamma_{\zeta}$, $\hat P(\hat\cW) > H^*(\hat \cW) + (1 + \zeta/2)\alpha$, where
  $ \hat \cW = \{x\in \cX \mid \hat H(x) > H^*(x)\}$. Since
  $|P(\hat\cW) - \hat P(\hat \cW)| \leq \zeta\alpha/4$, we have
  $P(\hat \cW) > H^*(\hat \cW) + (1 + \zeta/4)\alpha$, which is a contradiction to
  the assumption that $\tv(P, H^*) \leq \alpha$.
\end{proof}

\subsection{Obtaining a Semi-Agnostic Algorithm} \label{sec:semi}

Theorem~\ref{thm:tournament} shows that given a hypothesis class $\mathcal{H}$ and samples from an unknown distribution $P$, we can privately find a distribution $\hat{H} \in \mathcal{H}$ with $\dtv(P, \hat{H}) \le (3+\zeta)\alpha$ \emph{provided} that we know $\dtv(P, \mathcal{H}) \le \alpha$. But what if we are not promised that $P$ is itself close to $\mathcal{H}$? We would like to design a private hypothesis selection algorithm for the more general semi-agnostic setting, where for any value of $\operatorname{OPT} := \dtv(P, \mathcal{H})$, we are able to privately identify a distribution $\hat{H} \in \mathcal{H}$ with $\dtv(P, \hat{H}) \le c \cdot \operatorname{OPT} + \alpha$ for some universal constant $c$. Our goal will be to do this with sample complexity which is still logarithmic in $|\mathcal{H}|$.

Our strategy for handling this more general setting is by a reduction to that of Theorem~\ref{thm:tournament}. We run that algorithm $T = O(\log(1/\alpha))$ times, doubling the choice of $\alpha$ in each run and producing a sequence of candidate hypotheses $H_1, \dots, H_T$. By the guarantees of Theorem~\ref{thm:tournament}, there is some candidate $H_t$ with $\dtv(P, H_t) \le 2(3+\zeta)\operatorname{OPT}$. The remaining task is to approximately select the best candidate from $H_1, \dots, H_T$. This is done by implementing a private version of the Scheff\'{e} tournament which is itself semi-agnostic, but has a very poor (quadratic) dependence on the number of candidates $T$. 

We prove the following result, which gives a semi-agnostic learner whose sample complexity is comparable to that of Theorem~\ref{thm:tournament}.

\begin{theorem} \label{thm:what-did-i-do-to-deserve-this}
  Let $\alpha, \beta, \ve \in (0, 1)$, and $\zeta > 0$ be a constant. Let $\mathcal{H}$ be a set of $m$ distributions and let $P$ be a distribution with $\dtv(P, \mathcal{H}) = \operatorname{OPT}$. There is an $\ve$-differentially private algorithm which takes as input $n$ samples from $P$ and with probability at least $1-\beta$, outputs a distribution $\hat{H} \in \mathcal{H}$ with $\dtv(P, \hat{H}) \le 18(3+\zeta)\operatorname{OPT} + \alpha$, as long as
\[n \ge O\left( \frac{\log(m/\beta) + \log\log(1/\alpha)}{\alpha^2} + \frac{\log m + \log^2(1/\alpha) \cdot ( \log(1/\beta) + \log\log(1/\alpha))}{\alpha\ve} \right).\]
  The running time of the algorithm is $O(m^2 n \log (1/\alpha) +  n \log^2 (1/\alpha) )$.
\end{theorem}

As discussed above, the algorithm relies on the following variant with a much worse dependence on $m$.

\begin{lemma} \label{lem:terrible}
Let $\alpha, \beta, \ve \in (0, 1)$. There is an $\ve$-differentially private algorithm which takes as input $n$ samples from $P$ and with probability at least $1-\beta$, outputs a distribution $\hat{H} \in \mathcal{H}$ with $\dtv(P, \hat{H}) \le 9\operatorname{OPT} + \alpha$, as long as
\[n \ge O\left( \frac{\log(m/\beta)}{\alpha^2} + \frac{m^2\log(m/\beta)}{\alpha\ve} \right).\]
  The running time of the algorithm is $O(m^2 n)$.
\end{lemma}

\begin{proof}[Proof sketch.]
We use a different variation of the Scheff\'{e} tournament which appears in~\cite{DevroyeL01}. Non-privately, the algorithm works as follows. For every pair of hypotheses $H, H' \in \mathcal{H}$ with Scheff\'{e} set $\cW_{H, H'} = \{x\in \cX \mid H(x) > H'(x)\}$, let $H(\cW_{H, H'})$, $H'(\cW_{H, H'})$, and $P(\cW_{H, H'})$ denote the probability masses of $H, H', P$ on $\cW_{H, H'}$, respectively. Moreover, let $\hat{P}(\cW_{H, H'})$ denote the fraction of points in the input sample $D$ which lie in $\cW_{H, H'}$. We declare $H$ to be the winner of the pairwise contest between $H$ and $H'$ if $|H(\cW_{H, H'}) - \hat{P}(\cW_{H, H'})| < |H'(\cW_{H, H'}) - \hat{P}(\cW_{H, H'})|$. Otherwise, we declare $H'$ to be the winner. The algorithm outputs the hypothesis $\hat{H}$ which wins the most pairwise contests (breaking ties arbitrarily).

To make this algorithm $\ve$-differentially private, we replace $\hat{P}(\cW_{H, H'})$ in each pairwise contest with the $(\ve/{m \choose 2})$-differentially private estimate $c_{H, H'} = \hat{P}(\cW_{H, H'}) + \operatorname{Lap}({m\choose 2}/\ve n)$. By the composition guarantees of differential privacy, the algorithm as a whole is $\ve$-differentially private.

The analysis of Devroye and Lugosi~\cite[Theorem 6.2]{DevroyeL01} shows that the (private) Scheff\'{e} tournament outputs a hypothesis $\hat{H}$ with 
\[\dtv(\hat{H}, P) \le 9\operatorname{OPT} + 16 \max_{H, H' \in \mathcal{H}} \left| P(\cW_{H, H'}) - c_{H, H'} \right|.\]
Fix an arbitrary pair $H, H'$. A Chernoff bound shows that $|P(\cW_{H, H'}) - \hat{P}(\cW_{H, H'})| \le \alpha / 32$ with probability at least $1-\beta/(2m^2)$ as long as $n \ge O(\ln(m/\beta)/\alpha^2)$. Moreover, properties of the Laplace distribution guarantee $|c_{H, H'} - \hat{P}(\cW_{H, H'})| \le \alpha / 32$ with probability at least $1 - \beta/(2m^2)$ as long as $n \ge O(m^2\log(m/\beta)/\alpha\ve)$. The triangle inequality and a union bound over all pairs $H, H'$ complete the proof.
\end{proof}

\begin{proof}[Proof of Theorem~\ref{thm:what-did-i-do-to-deserve-this}]
We now combine the private hypothesis selection algorithm of Theorem~\ref{thm:tournament} with the expensive semi-agnostic learner of Lemma~\ref{lem:terrible} to prove Theorem~\ref{thm:what-did-i-do-to-deserve-this}. Define sequences $\alpha_1 = \alpha/126, \alpha_2 = 2\alpha/126, \dots, \alpha_T = 2^{T-1}\alpha/126$ and $\ve_1 = \ve / 4, \ve_2 = \ve / 8, \dots, \ve_T = 2^{-(T+1)}\ve$ for $T = \lceil \log_2 (1/\alpha) \rceil + 1$. For each $t = 1, \dots, T$, let $H_t$ denote the outcome of a run of Algorithm~\ref{alg:phs} using accuracy parameter $\alpha_t$ and privacy parameter $\ve_t$. Finally, use the algorithm of Lemma~\ref{lem:terrible} to select a hypothesis from $H_0, \dots, H_T$ using accuracy parameter $\alpha$ and privacy parameter $\ve / 2$.

Privacy of this algorithm follows immediately from composition of differential privacy. We now analyze its sample complexity guarantee. By Lemma~\ref{lem:tournament-utility}, we have that all $T$ runs of Algorithm~\ref{alg:phs} succeed simultaneously with probability at least $1-\beta/2$ as long as
\[n \ge O\left( \frac{\log(m/\beta) + \log\log(1/\alpha)}{\alpha^2} + \frac{\log(m/\beta) + \log\log(1/\alpha)}{\alpha\ve}\right).\]
  Condition on this event occurring. Recall that success of run $t$ of Algorithm~\ref{alg:phs} means that if $\operatorname{OPT} \in (\alpha_{t-1}, \alpha_t]$, then $\dtv(P, H_t) \le (3+\zeta)\alpha_t \le 2(3+\zeta)\operatorname{OPT}$. Meanwhile, if $\operatorname{OPT} \le \alpha_1 = \alpha/126$, then we have $\dtv(P, H_1) \le \alpha/18$. Hence, regardless of the value of $\operatorname{OPT}$, there exists a run $t$ such that $\dtv(P, H_t) \le 2(3 +\zeta) \operatorname{OPT} + \alpha/18$. The algorithm of Lemma~\ref{lem:terrible} is now, with probability at least $1-\beta/2$, able to select a hypothesis $\hat{H}$ with $\dtv(P, \hat{H}) \le 9\dtv(P, H_t) + \alpha / 2 \le 18(3+\zeta)\operatorname{OPT} + \alpha$ as long as
\[n  \ge O\left( \frac{\log(1/\beta)  +\log\log(1/\alpha)}{\alpha^2} + \frac{\log^2(1/\alpha) \cdot \left(\log(1/\beta) + \log\log(1/\alpha)\right)}{\alpha\ve} \right).\]
This gives the asserted sample complexity guarantee.
\end{proof}



\section{An Advanced Method for Private Hypothesis Selection}
\label{sec:infinite}

In Section \ref{sec:hyp-sel}, we provided a simple algorithm whose sample complexity grows logarithmically in the size of the hypothesis class. We now demonstate that this dependence can be improved and, indeed, we can handle infinite hypothesis classes given that their VC dimension is finite and that the cover has small doubling dimension.

To obtain this improved dependence on the hypothesis class size, we must make two improvements to the analysis and algorithm.
First, rather than applying a union bound over all the pairwise contests to analyse the tournament, we use a uniform convergence bound in terms of the VC dimension of the Scheff\'e sets. 
Second, rather than use the exponential mechanism to select a hypothesis, we use a ``GAP-MAX'' algorithm~\cite{BunDRS18}. This takes advantage of the fact that, in many cases, even for infinite hypothesis classes, only a handful of hypotheses will have high scores. The GAP-MAX algorithm need only pay for the hypotheses that are close to optimal.
To exploit this, we must move to a relaxation of pure differential privacy which is not subject to strong packing lower bounds (as we describe in Section \ref{sec:packing}). 
Specifically, we consider approximate differential privacy, although results with an improved dependence are also possible under various variants of concentrated differential privacy \cite{DworkR16,BunS16,Mironov17,BunDRS18}.

\begin{theorem}\label{thm:sel-inf}
	Let $\mathcal{H}$ be a set of probability distributions on $\mathcal{X}$.  
	Let $d$ be the VC dimension of the set of functions $f_{H,H'} : \mathcal{X} \to \{0,1\}$ defined by $f_{H,H'}(x)=1\iff H(x)>H'(x)$ where $H,H'\in\mathcal{H}$.
	There exists a $(\varepsilon,\delta)$-differentially private algorithm which has following guarantee. Let $D = \{X_1, \dots, X_n\}$ be a set of private samples
	drawn independently from an unknown probability distribution $P$.
	Let $k = \left|\left\{ H \in \mathcal{H} : \tv(H,P) \le 7\alpha\right\}\right|$. 
	Suppose there exists a distribution $H^* \in \mathcal{H}$ such that $\tv(P, H^*) \leq \alpha$.  If $n = \Omega\left(\frac{d + \log(1/\beta)}{\alpha^2} + \frac{\log(k/\beta) +\min\{\log|\mathcal{H}|,\log(1/\delta)\}}{\alpha
		\ve}\right)$, then the algorithm will output a distribution
	$\hat H \in \mathcal{H}$ such that $\dtv(P, \hat H) \leq 7\alpha$ with probability at least $1-\beta$.
	
	Alternatively, we can demand that the algorithm be $\frac12\varepsilon^2$-concentrated differentially private if $n = \Omega\left(\frac{d + \log(1/\beta)}{\alpha^2} + \frac{\log(k/\beta) +\sqrt{\log|\mathcal{H}|}}{\alpha
		\ve}\right)$.
\end{theorem}

Comparing Theorem \ref{thm:sel-inf} to Theorem \ref{thm:tournament}, we see that the first (non-private) $\log|\mathcal{H}|$ term is replaced by the VC dimension $d$ and the second (private) $\log|\mathcal{H}|$ term is replaced by $\log k + \log(1/\delta)$.  
Here $k$ is a measure of the ``local'' size of the hypothesis class $\mathcal{H}$; its definition is similar to that of the doubling dimension of the hypothesis class under total variation distance. 

We note that the $\log(1/\delta)$ term could be large, as the privacy failure probability $\delta$ should be cryptographically small. 
Thus our result includes statements for pure differential privacy (by using the other term in the minimum with $\delta = 0$) and also concentrated differential privacy.
Note that, since $d$ and $\log k$ can be upper-bounded by $O(\log |\mathcal{H}|)$, this result supercedes the guarantees of Theorem~\ref{thm:tournament}.

\subsection{VC Dimension}

We begin by reviewing the definition of Vapnik-Chervonenkis (VC) dimension and its properties.

\begin{definition}[VC dimension \cite{VapnikC74}]\label{defn:VC}
	Let $\mathcal{F}$ be a set of functions $f : \mathcal{X} \to \{0,1\}$. The VC dimension of $\mathcal{F}$ is defined to be the largest $d$ such that there exist $x_1, \cdots, x_d \in \mathcal{X}$ and $f_1, \cdots, f_{2^d} \in \mathcal{H}$ such that for all $1 \le i < j \le 2^d$ there exists $1 \le k \le d$ such that $f_i(x_k) \ne f_j(x_k)$.
\end{definition}

For our setting, we must extend the definition of VC dimension from function families to hypothesis classes.

\begin{definition}[VC dimension of hypothesis class]
	Let $\mathcal{H}$ be a set of probability distributions on a space $\mathcal{X}$. For $H,H'\in\mathcal{H}$, define $f_{H,H'} : \mathcal{X} \to \{0,1\}$ by $f(x)=1 \iff H(x)>H'(x)$. Define $\mathcal{F}(\mathcal{H})=\{f_{H,H'} : H,H'\in\mathcal{H}\}$. We define the VC dimension of $\mathcal{H}$ to be the VC dimension of $\mathcal{F}(\mathcal{H})$.\footnote{Here, for simplicity, we assume that each distribution $H$ is given by a density function $H(\cdot)$. More generally, we define the VC dimension of $\mathcal{H}$ to be the smallest $d$ such that there exists a function family $\mathcal{F} \subseteq \{0,1\}^{\mathcal{X}}$ of VC dimension $d$ with the property that, for all $H,H'\in\mathcal{H}$ we have $\tv(H,H') = \sup_{f \in \mathcal{F}} \ex[X\leftarrow H]{f(X)}-\ex[X\leftarrow H']{f(X)}$, where the supremum is over $f$ measurable with respect to both $H$ and $H'$. We ignore this technicality throughout.}
\end{definition}

The key property of VC dimension is the following uniform convergence bound, which we use in place of a union bound.

\begin{theorem}[Uniform Convergence \cite{Talagrand94}]\label{thm:VC-unif_conv}
	Let $\mathcal{F}$ be a set of functions $f : \mathcal{X} \to \{0,1\}$ with VC dimension $d$. Let $P$ be a distribution on $\mathcal{X}$. Then $$\pr[D \leftarrow P^n]{\sup_{f \in \mathcal{F}} |f(D)-f(P)| \le \alpha} \ge 1-\beta$$ whenever $n = \Omega\left(\frac{d+\log(1/\beta)}{\alpha^2}\right)$. Here $f(D) := \frac{1}{n} \sum_{x \in D} f(x)$ and $f(P):=\ex[X \leftarrow P]{f(X)}$.
\end{theorem}

It is immediate from Definition \ref{defn:VC} that $VC(\mathcal{F})\le\lfloor \log_2|\mathcal{F}| \rfloor $. Thus Theorem \ref{thm:VC-unif_conv} subsumes the union bound used in the proof of Theorem~\ref{thm:tournament}.

The relevant application of uniform convergence for our algorithm is the following lemma (roughly the equivalent of Lemma~\ref{friday}), which says that good hypotheses have high scores, and bad hypotheses have low scores.

\begin{lemma}\label{lem:sel-unif-conv}
	Let $\mathcal{H}$ be a collection of probability distributions on $\mathcal{X}$  with VC dimension $d$.
	
	Let $S : \mathcal{H} \times \mathcal{X}^n \to \mathbb{R}$ be a score function similar to \eqref{score}, namely $$S(H,D) =\inf_{H'\in \mathcal{H}}  \max\left\{\begin{array}{c} |\{x \in D : H(x)>H'(x)\}|-n\cdot(\pr[X\leftarrow H']{H(X)>H'(X)}+3\alpha),\\n \cdot \mathbb{I}[\tv(H,H')\le6\alpha]\end{array}\right\},$$
	where $\mathbb{I}$ denotes the indicator function.
	
	Let $P$ be a distribution on $\mathcal{X}$. Let $\alpha,\beta>0$ and $n \ge O(\frac{1}{\alpha^2}(d+\log(1/\beta)))$. Suppose there exists $H^*\in \mathcal{H}$ with $\tv(P, H^*) \leq \alpha$. Then, with probability at least $1-\beta$ over $D \leftarrow P^n$, we have
	\begin{itemize}
		\item $S(H^*,D)>\alpha n$ and
		\item $S(H,D)=0$ for all $H \in \mathcal{H}$ with $\tv(H,P)>7\alpha$.
	\end{itemize}
\end{lemma}
\begin{proof}
	For $H,H'\in\mathcal{H}$, define $f_{H,H'} : \mathcal{X} \to \{0,1\}$ by $f_{H,H'}(x)=1 \iff H(x)>H'(x)$.
	Note that $|\{x \in D : H(x)>H'(x)\}| = \sum_{x \in D} f_{H,H'}(x)$ and $d$ is the VC dimension of the function class $\{f_{H,H'} : H,H'\in\mathcal{H}\}$. By Theorem~\ref{thm:VC-unif_conv}, if $n = \Omega\left(\frac{d+\log(1/\beta)}{\alpha^2}\right)$, then $$\pr[D \leftarrow P^n]{\forall H,H'\in\mathcal{H}~~\left||\{x \in D : H(x)>H'(x)\}| - n \cdot \pr[X \leftarrow P]{H(X)>H'(X)}\right| \le \alpha n} \ge 1-\beta.$$
	We condition on this event happening. 
	
	In order to prove the first conclusion -- namely, $S(H^*,D)>\alpha n$ -- it remains to show that, for all $H' \in \mathcal{H}$, we have either $\tv(H^*,H') \le 6\alpha$ or $$|\{x \in D : H(x)>H'(x)\}|-n\cdot(\pr[X\leftarrow H']{H^*(X)>H'(X)}+3\alpha)>\alpha n.$$ If $\tv(H^*,H') \le 6\alpha$, we are done, so assume $\tv(H^*,H')> 6\alpha$.  By the uniform convergence event we have conditioned on, 
	\begin{align*}
		|\{x \in D : H(x)>H'(x)\}|&\ge n \cdot (\pr[X \leftarrow P]{H(X)>H'(X)}-\alpha)\\
		&\ge n \cdot (\pr[X \leftarrow H^*]{H(X)>H'(X)}-\tv(P,H^*)-\alpha)\\
		&\ge n\cdot(\tv(H^*,H')+\pr[X \leftarrow H']{H(X)>H'(X)}-2\alpha)\\
		&> n\cdot(6\alpha+\pr[X \leftarrow H']{H(X)>H'(X)}-2\alpha),
	\end{align*}
	from which the desired conclusion follows.
	
	In order to prove the second conclusion -- namely, $S(H,D)=0$ for all $H \in \mathcal{H}$ with $\tv(H,P)>7\alpha$ -- it suffices to show that one $H'\in\mathcal{H}$ yields a score of zero for any $H \in \mathcal{H}$ with $\tv(H,P)>7\alpha$. In particular, we show that $H'=H^*$ yields a score of zero for any such $H$. That is, if $\tv(H,P)>7\alpha$, then $\tv(H,H^*)> 6\alpha$ and $$|\{x \in D : H(x)>H^*(x)\}|-n\cdot(\pr[X\leftarrow H^*]{H(X)>H^*(X)}+3\alpha) \le 0.$$ By the triangle inequality $\tv(H,H^*) \ge \tv(H,P)-\tv(P,H^*) > 7\alpha-\alpha=6\alpha$, as required. By the uniform convergence event we have conditioned on, 
	\begin{align*}
	|\{x \in D : H(x)>H^*(x)\}|&\le n \cdot (\pr[X \leftarrow P]{H(X)>H^*(X)}+\alpha)\\
	&\le n \cdot (\pr[X \leftarrow H^*]{H(X)>H^*(X)}+\tv(P,H^*)+\alpha)\\
	&\le n\cdot(\pr[X\leftarrow H^*]{H(X)>H^*(X)}+2\alpha),
	\end{align*}
	which completes the proof.
\end{proof}

\subsection{GAP-MAX Algorithm}

In place of the exponential mechanism for privately selecting a hypothesis we use the following algorithm that works under a ``gap'' assumption. That is, we assume that there is a $5\alpha n$ gap between the highest score and the $(k+1)$-th highest score. Rather than paying in sample complexity for the total number of hypotheses we pay for the number of high-scoring hypotheses $k$.

 This algorithm is based on the GAP-MAX algorithm of Bun, Dwork, Rothblum, and Steinke \cite{BunDRS18}. However, we combine their GAP-MAX algorithm with the exponential mechanism to improve the dependence on the parameter $k$.

\begin{theorem}\label{thm:gap-max}
	Let $\mathcal{H}$ and $\mathcal{X}$ be arbitrary sets. Let $S : \mathcal{H} \times \mathcal{X}^n \to \mathbb{R}$ have sensitivity at most 1 in its second argument -- that is, for all $H \in \mathcal{H}$ and all $D,D'\in\mathcal{X}^n$ differing in a single example, $|S(H,D)-S(H,D')|\le 1$. 
	
	For $D \in \mathcal{X}^n$ and $\alpha>0$, define $$K(D,5\alpha):=\left|\left\{H \in \mathcal{H} : S(H,D) \ge \sup_{H' \in \mathcal{H}} S(H',D)-5\alpha n\right\}\right|.$$
	
	Given parameters $\varepsilon,\delta,\beta>0$ and $n,k \ge 1$, there exists a $(\varepsilon,\delta)$-differentially private randomized algorithm $M : \mathcal{X}^n \to \mathcal{H}$ such that, for all $D \in \mathcal{X}^n$ and all $\alpha>0$, $$K(D,5\alpha) \le k \implies \pr{S(M(D),D) \ge \sup_{H' \in \mathcal{H}} S(H',D) - \alpha n} \ge 1-\beta$$ provided $n = \Omega\left(\frac{\min\{\log |\mathcal{H}|,\log(1/\delta)\}+\log(k/\beta)}{\alpha\varepsilon}\right)$.
	
	Furthermore, given $\varepsilon,\beta>0$ and $n,k\ge 1$, there exists a $\frac12 \varepsilon^2$-concentrated differentially private \cite{BunS16} algorithm $M : \mathcal{X}^n \to \mathcal{H}$ such that, for all $D \in \mathcal{X}^n$ and all $\alpha>0$, $$K(D,5\alpha) \le k \implies \pr{S(M(D),D) \ge \sup_{H' \in \mathcal{H}} S(H',D) - \alpha n} \ge 1-\beta$$ provided $n = \Omega\left(\frac{\sqrt{\log |\mathcal{H}|}+\log(k/\beta)}{\alpha\varepsilon}\right)$.
\end{theorem}
\begin{proof}
	We begin by describing the algorithm.
	\begin{enumerate}
		\item Let $m = \left\lceil\frac{k^2}{\beta}\right\rceil$ and let $G : \mathcal{H} \to [m]$ be a uniformly random function.\footnote{It suffices for $G$ to be a drawn from a universal hash function family.} 
		\item Randomly select $B \in [m]$ with $$\pr{B=b} \propto \exp\left(\frac{\varepsilon}{4} \sup \left\{ S(H,D) : H \in \mathcal{H}, G(H)=b\right\}\right).$$
		\item Define $\mathcal{H}_B = \{H \in \mathcal{H} : G(H)=B\}$. Let $H_B^1 = \argmax_{H \in \mathcal{H}_B} S(H,D)$ and $H_B^2 = \argmax_{H \in \mathcal{H}_B \setminus \{H_B^1\}} S(H,D)$, breaking ties arbitrarily. (That is, $\mathcal{H}_B$ is the $B$-th ``bin'' and $H_B^1$ and $H_B^2$ are the items in this bin with the largest and second-largest scores respectively.) Define $S_B' : \mathcal{H}_B \times \mathcal{X}^n \to \mathbb{R}$ by $$S_B'(H,D) = \frac12 \max\{0,S(H,D)-S(H_B^2,D)\}.$$ (Note that $S_B'$ has sensitivity 1 and $S_B'(H,D)=0$ whenever $H \ne H_B^1$.)
		\item Let $\mathcal{D}$ be a distribution on $\mathbb{R}$ such that adding a sample from $\mathcal{D}$ to a sensitivity-$1$ function provides $(\varepsilon/4,\delta/2)$-differential privacy (or, respectively, $\frac{1}{6}\varepsilon^2$-concentrated differential privacy). For example, $\mathcal{D}$ could be a Laplace distribution with scale $4/\varepsilon$ truncated to the interval $[-t,t]$ for $t=4(1+\log(1/\delta))/\varepsilon$ (or unbounded if $\delta=0$). To attain concentrated differential privacy, we can set $\mathcal{D}=N\left(0,\frac{3}{\varepsilon^2}\right)$, a centered Gaussian with variance $3/\ve^2$.
		\item Draw a sample $Z_H$ i.i.d. from $\mathcal{D}$ corresponding to every $H \in \mathcal{H}_B$.
		\item Return $H^* = \argmax_{H \in \mathcal{H}_B} S_B'(H,D) + Z_H$. 
	\end{enumerate}
	The selection of $B$ is an instantiation of the exponential mechanism \cite{McSherryT07} and is $(\varepsilon/2,0)$-differentially private. The selection of $H^*$ in the final step is a GAP-MAX algorithm \cite{BunDRS18} and is $(\varepsilon/2,\delta)$-differentially private.  By composition, the entire algorithm is $(\varepsilon,\delta)$-differentially private (or, respectively, $\frac12\varepsilon^2$-concentrated differentially private).
	
	For the utility analysis, in order for the algorithm to output a good $H^*$, it suffices for the following three events to occur.
	\begin{itemize}
		\item $S(H_B^1,D) \ge \sup_{H' \in \mathcal{H}} S(H',D) - \alpha n$.\\That is, restricting the search to $\mathcal{H}_B$, rather than all of $\mathcal{H}$, only reduces the score of the optimal choice by $ \alpha n$. The exponential mechanism ensures that this happens with probability at least $1-\beta/4$, as long as $n \ge \frac{4\log(2 k/\beta)}{\varepsilon\alpha}$.
		\item $S(H_B^2,D) <  \sup_{H' \in \mathcal{H}} S(H',D) -  5\alpha n$.\\That is, the second-highest score within $\mathcal{H}_B$ is at least $5\alpha n$ less than the highest score overall. We have assumed that there are at most $k$ elements $H \in \mathcal{H}$ such that $S(H,D) \ge  \sup_{H' \in \mathcal{H}} S(H',D) -  5\alpha n$. Call these ``large elements.'' Since $G : \mathcal{H} \to [m]$ is random and $m \ge k^2/\beta$, the probability that an arbitrary but fixed pair of large elements collide -- that is, are in the same $H_B$ is $1/m \le \beta/k^2$. If we union bound over the ${k \choose 2} < k^2/2$ pairs, we see that the probability of any collisions is at most $\beta/2$. Thus, the probability that more than one large element satisfies $G(H)=B$ is at most $\beta/2$. This suffices for the event to occur.
		\item $\sup_{H\in \mathcal{H}_B} |Z_H| \le  \alpha n$.\\If the noise distribution $\mathcal{D}$ is supported on $[-\alpha n, \alpha n]$, then this condition holds with probability $1$. For the truncated Laplace distribution, this is possible whenever $ n \ge 1+4\log(1/\delta)/\alpha\varepsilon$. Alternatively, we can use unbounded Laplace noise and a union bound to show that this event occurs with probability at least $1-\beta/4$ whenever $n \ge 4\log(4|\mathcal{H}_B|/\beta)/\varepsilon\alpha$. For Gaussian noise, $n \ge \frac{3}{\varepsilon\alpha}\sqrt{\log(4|\mathcal{H}_B|/\beta)}$ suffices.
	\end{itemize}
	Assuming the first and second events occur, we have $S'_B(H_B^1,D) = \frac{S(H_B^1,D) - S(H_B^2,D)}{2} > 2\alpha n$. Given this, the third event implies $H^*=H_B^1$. Finally, the first event then implies $S(H^*, D) \ge \sup_{H' \in \mathcal{H}} S(H',D) - \alpha n$, as required. A union bound over the three events completes the proof.
\end{proof}

Now we can combine the VC-based uniform convergence bound with the GAP-MAX algorithm to prove our result.

\begin{proof}[Proof of Theorem \ref{thm:sel-inf}]
	By Lemma \ref{lem:sel-unif-conv}, with high probability over the draw of the dataset $D$, our score function satisfies $\sup_{H \in \mathcal{H}} S(H,D) \ge S(H^*,D)>\alpha n$ and $S(H,D)=0$ whenever $\tv(H,P)>7\alpha$. This requires $n = \Omega(d/\alpha^2)$.
	
	Note that the score function $S$ has sensitivity-1, since it is the supremum of counts. Conditioned on the uniform convergence event, the maximum score is at least $\alpha n$ and there are at most $k$ elements of $\mathcal{H}$ with score greater than $0$. Thus we can apply the GAP-MAX algorithm of Theorem \ref{thm:gap-max}. If $n = \Omega((\min\{\log |\mathcal{H}|, \log(1/\delta)\}+\log(k))/\alpha\varepsilon)$, then with high probability, the algorithm outputs $\hat H \in \mathcal{H}$ with score at least $\frac45 \alpha n$, as required.
\end{proof}


\section{Packings, Lower Bounds, and Relations to Covers}
\label{sec:packing}
In this section, we show that the sample complexity of our algorithms for private hypothesis selection with pure differential privacy cannot be improved, at least for constant values of the proximity parameter $\alpha$. We first apply a packing argument~\cite{HardtT10, BeimelBKN14} to show a lower bound which is logarithmic in the packing number of the class of distributions (Lemma~\ref{lem:packing}). We then state a folklore relationship between the sizes of maximal packings and minimal covers (Lemma~\ref{lem:pack-n-cover}), which shows that instantiating our private hypothesis selection algorithm with a minimal cover gives essentially optimal sample complexity (Theorem~\ref{thm:packing-tight}).



\begin{lemma}
	\label{lem:packing}
	Suppose there exists an $\alpha$-packing $\mathcal{P}_{\alpha}$ of a set of distributions $\mathcal{H}$.
	Then any $\ve$-differentially private algorithm which takes as input samples $X_1, \dots, X_n \sim P$ for some $P \in \mathcal{H}$ and produces a distribution $\hat H$ such that $\dtv(P,\hat H) \leq \alpha/2$ with probability $\geq 9/10$ requires
	$$n = \Omega\left(\frac{\log |\mathcal{P}_{\alpha}|}{\varepsilon}\right).$$
\end{lemma}

One might conjecture a stronger version of this lemma exists, and that one could prove the lower bound $n = \Omega\left(\frac{\log |\mathcal{P}_{\alpha}|}{\alpha\varepsilon}\right)$.
However, such a statement cannot be true in general.
For example, consider an $\alpha$-packing of $N(\mu, 1)$ where $\mu \in [-R, R]$, which would have size $\Omega(R/\alpha)$. 
If such a lemma were true, it would imply a lower bound of $\tilde\Omega\left(\frac{\log R}{\alpha \ve}\right)$, which contradicts known upper bounds. 
One way to prove a lower bound achieving such a dependence on $\alpha$ would be to have a single ``central'' distribution which is close to all distributions in the packing (see an argument of this sort in Theorem 5.13 of~\cite{Vadhan17}).
However, we do not explore this here, as our goal is to match our upper bound which is stated in terms of a generic cover.

\begin{proof}
	Let $M$ be a $\ve$-differentially private algorithm with the stated accuracy requirement, and denote by $M(P^n)$ the distribution on hypotheses obtained by running $M$ on $n$ i.i.d.\ samples from a distribution $P \in \mathcal{H}$. For each $P \in \mathcal{P}_\alpha$, let $B_P$ denote the set of distributions which are at total variation distance at most $\alpha/2$ from $P$. Then the accuracy requirement implies that $\pr[\hat{H} \gets M(P^n)]{\hat{H} \in B_P} \ge 9/10$ for all $P \in \mathcal{H}$. Let $P_0 \in \mathcal{P}_\alpha$ be an arbitrary packing element. 
	Note that, trivially, samples from $P^n$ and $P^n_0$ have Hamming distance at most $n$ for any $P$.
	Recall the group privacy property of differential privacy, which states that if $M$ is $\ve$-DP, then $\Pr[M(X) \in S] \leq \exp\left(\ve d(X,X')\right) \cdot \Pr[M(X') \in S]$ for any set $S \subseteq \mathrm{Range}(M)$, where $d(X,X')$ is the Hamming distance between the two datasets.
	Applying this property with $P^n$ and $P^n_0$, we have
	\[\pr[\hat{H} \gets M(P_0^n)]{\hat{H} \in B_P} \ge e^{-\ve n} \cdot 9/10\]
	for every $P \in \mathcal{P}_\alpha$. The fact that $\mathcal{P}_\alpha$ is an $\alpha$-packing implies that the sets $B_P$ are all disjoint, and hence
	\[1 \ge \sum_{P \in \mathcal{P}_\alpha} \pr[\hat{H} \gets M(P_0^n)]{\hat{H} \in B_P} \ge \left| \mathcal{P}_\alpha \right| \cdot e^{-\ve n} \cdot 9/10.\]
	Rearranging gives us the stated lower bound on $n$.
\end{proof}

The following lemma is a well-known folklore relationship between packing and covering numbers.
We include a proof for completeness.

\begin{lemma}
	\label{lem:pack-n-cover}
	For a set of distributions $\mathcal{H}$, let $p_\alpha$ and $c_\alpha$ be the size of the largest $\alpha$-packing and smallest $\alpha$-cover of $\mathcal{H}$, respectively.
	Then
	$$p_{2\alpha} \leq c_\alpha \leq p_{\alpha}.$$
\end{lemma}
\begin{proof}
	We first prove the inequality on the left. Let $\mathcal{C}_\alpha$ be a cover of $\mathcal{H}$ of minimal size $c_\alpha$. If $c_\alpha = \infty$, we are done. Otherwise, let $S \subseteq \mathcal{H}$ be any set of distributions of size at least $c_\alpha + 1$. By the pigeonhole principle, there exists $P \in \mathcal{C}_\alpha$ and two distinct distributions $Q, Q' \in S$ such that $\dtv(P, Q) \le \alpha$ and $\dtv(P, Q') \le \alpha$.  Hence $\dtv(Q, Q') \le 2\alpha$ by the triangle inequality, so $S$ cannot be $(2\alpha)$-packing of $\mathcal{H}$. This suffices to show that $p_{2\alpha} \le c_\alpha$. 
	
	Next, we prove the inequality on the right. Let $\mathcal{P}_{\alpha}$ be a maximal $\alpha$-packing with size $|\mathcal{P}_{\alpha}| = p_{\alpha}$. If $p_{\alpha} = \infty$, we are done. Otherwise, we claim that $\mathcal{P}_{\alpha}$ is also an $\alpha$-cover of $\mathcal{H}$, and hence $c_\alpha \le |\mathcal{P}_\alpha| = p_\alpha$. To see this, suppose for the sake of contradiction that there were a distribution $P \in \mathcal{H}$ with $\dtv(P, \mathcal{P}_{\alpha}) > \alpha$. Then we could add $P$ to $\mathcal{P}_{\alpha}$ to produce a strictly larger packing, contradicting the maximality of $\mathcal{P}_{\alpha}$.
\end{proof}

\begin{theorem}
	\label{thm:packing-tight}
	Let $\mathcal{H}$ be a set of distributions, and let $n^*_\alpha$ denote the minimum number of samples such that there exists an $\ve$-differentially private algorithm which takes as input samples $X_1, \dots, X_{n^*_\alpha} \sim P$ for an arbitrary $P \in \mathcal{H}$ and outputs a distribution $\hat H$ such that $\dtv(P,\hat H) \leq \alpha/2$ with probability $\geq 9/10$. Then there exists a cover of $\mathcal{H}$ such that the instantiation of the algorithm underlying Theorem~\ref{thm:tournament} with this cover takes as input $n = \Omega(n^*_\alpha \cdot (\ve/\alpha^2 + 1/\alpha))$ samples from an arbitrary $P \in \mathcal{H}$ and outputs a $\hat H$ such that $\dtv(P,\hat H) \leq (6+2\zeta) \alpha$ with probability $\geq 9/10$ for any constant $\zeta > 0$.
\end{theorem}

\begin{proof}
	Let $p_\alpha$ denote the size of the largest $\alpha$-packing of $\mathcal{H}$. By Lemma~\ref{lem:packing}, we have $n^*_\alpha = \Omega(\log p_\alpha / \ve)$. On the other hand, by Lemma~\ref{lem:pack-n-cover}, we know that there exists an $\alpha$-cover $\mathcal{C}_\alpha$ of $\mathcal{H}$ with $|\mathcal{C}_\alpha| \le p_\alpha$. Hence $\log |\mathcal{C}_\alpha| \le O(\ve \cdot n^*_\alpha)$ and the asserted sample complexity guarantee follows from Corollary~\ref{cor:covertolearn}.
\end{proof}


\section{Applications of Hypothesis Selection}
\label{sec:apps}
In this section, we give a number of applications of Theorem~\ref{thm:tournament}, primarily to obtain sample complexity bounds for learning a number of distribution classes of interest.
Recall Corollary~\ref{cor:covertolearn}, which is an immediate corollary of Theorem~\ref{thm:tournament}. 
This indicates that we can privately semi-agnostically learn a class of distributions with a number of samples proportional to the logarithm of its covering number.
\covertolearn*

Note that the factor of $(6+2\zeta)\alpha$ in the corollary statement (versus $(3+\zeta)\alpha$ in the statement of Theorem~\ref{thm:tournament}) is due to the fact the algorithm is semi-agnostic, and the closest element in the cover is $2\alpha$-close to $P$, rather than just $\alpha$-close.

We instantiate this result to give the sample complexity results for semi-agnostically learning product distributions (Section~\ref{sec:prod-app}), Gaussian distributions (Section~\ref{sec:gaussian-app}), sums of some independent random variable classes (Section~\ref{sec:sirv-app}), piecewise polynomials (Section~\ref{sec:pp-app}), and mixtures (Section~\ref{sec:mix-app}).
Furthermore, we mention an application to private PAC learning (Section~\ref{sec:pac-app}), when the distribution of unlabeled examples is known to come from some hypothesis class.

\subsection{Product Distributions}
\label{sec:prod-app}
As a first application, we first give an $\ve$-differentially private algorithm for learning product distributions over discrete alphabets.

\begin{definition}
  A \emph{$(k,d)$-product distribution} is a distribution over $[k]^d$, such that its marginal distributions are independent (i.e., the distribution is the product of its marginals).
\end{definition}

We start by constructing a cover for product distributions.
\begin{lemma}
  There exists an $\alpha$-cover of the set of $(k,d)$-product distributions of size
  $$O\left(\frac{kd}{\alpha}\right)^{d(k-1)}.$$
\end{lemma}
\begin{proof}
  Consider some fixed product distribution $P$, with marginal distributions $(P_1, \dots, P_d)$.
  We will construct a cover that contains a distribution $Q$ (with marginals $(Q_1, \dots, Q_d)$) that is $\alpha$-close in total variation distance.
  
  First, by triangle inequality, we have that $\dtv(P,Q) \leq \sum_{i=1}^d \dtv(P_i, Q_i)$, so it suffices to approximate each marginal distribution to accuracy $\alpha/d$.
  Stated another way, we must generate an $(\alpha/d)$-cover of distributions over $[k]$, and we can then take its $d$-wise Cartesian product.
  Raising the size of this underlying cover to the power $d$ gives us the size of the overall cover.

  To $(\alpha/d)$-cover a distribution over $[k]$, we will additively grid the probability of each symbol at granularity $\Theta\left(\frac{\alpha}{kd}\right)$, choosing the probability of the last symbol $k$ such that the sum is normalized.
  This will incur $\Theta\left(\frac{\alpha}{kd}\right)$ error per symbol (besides for symbol $k$), and summing over the $k-1$ symbols accumulates error $\Theta\left(\frac{\alpha}{d}\right)$.
  It can also be argued that the error on symbol $k$ is $O\left(\frac{\alpha}{d}\right)$ -- with an appropriate choice of granularity, this gives us an $(\alpha/d)$-cover for distributions over $[k]$.
  The size of this cover is $O\left(\frac{kd}{\alpha}\right)^{k-1}$, which allows us to conclude the lemma statement.
\end{proof}

With this cover in hand, applying Corollary~\ref{cor:covertolearn} allows us to conclude the following sample complexity upper bound.
\begin{corollary}
\label{cor:learn-prod}
  Suppose we are given a set of samples $X_1, \dots, X_n \sim P$, where $P$ is $\alpha$-close to a $(k,d)$-product distribution.
  Then for any constant $\zeta > 0$, there exists an $\ve$-differentially private algorithm which outputs a $(k,d)$-product distribution $H^*$ such that $\dtv(P,H^*) \leq (6+2\zeta)\alpha$ with probability $\geq 9/10$, so long as
  $$n = \Omega\left(kd \log\left(\frac{kd}{\alpha}\right)\left(\frac{1}{\alpha^2} + \frac{1}{\alpha \ve}\right)\right).$$
\end{corollary}

This gives the first $\tilde O(d)$ sample algorithm for learning a binary product distribution in total variation distance under pure differential privacy, improving upon the work of Kamath, Li, Singhal, and Ullman~\cite{KamathLSU19} by strengthening the privacy guarantee at a minimal cost in the sample complexity.
The natural way to adapt their result from concentrated to pure differential privacy would require $\Omega(d^{3/2})$ samples.

\begin{remark}\label{rem:separation}
	Properly learning a product distribution over $\{0,1\}^d$ to total variation distance $\le\frac12$ implies learning its mean $\mu \in [0,1]^d$ up to $\ell_1$ error $\le 2\sqrt{d}$; see Lemma \ref{lem:prod-tv-L1} below.
	
	Thus Corollary \ref{cor:learn-prod} implies a $\varepsilon$-differentially private algorithm which takes $n=\tilde O(d/\varepsilon)$ samples from a product distribution $P$ on $\{0,1\}^d$ and, with high probability, outputs an estimate $\hat \mu$ of its mean $\mu$ with $\|\hat \mu -\mu \|_1 \le 2\sqrt{d}$.
	
	In contrast, for non-product distributions over the hypercube, estimating the mean to the same accuracy under $\varepsilon$-differential privacy requires $n=\Omega(d^{3/2}/\varepsilon)$ samples \cite{HardtT10,SteinkeU15}. Thus we have a polynomial separation between estimating product and non-product distributions under pure differential privacy.
\end{remark}

\begin{lemma}\label{lem:prod-tv-L1}
	If $P$ and $Q$ are product distributions on $\mathbb{R}^d$ with $\tv(P,Q)\le\frac12$ and per-coordinate variance at most $\sigma^2$, then $$\left\|\ex[X\leftarrow P]{X}-\ex[X\leftarrow Q]{X}\right\|_1 \le 4\sqrt{d\sigma^2}.$$
\end{lemma}
\begin{proof}
	Let $\mu=\ex[X\leftarrow P]{X} \in \mathbb{R}^d$ and $\mu'=\ex[X\leftarrow Q]{X} \in \mathbb{R}^d$. Let $\tau=\|\mu-\mu'\|_1$. Let $\nu = \mathrm{sign}(\mu-\mu') \in \{-1,+1\}^d$ so that $\langle \nu , \mu-\mu' \rangle = \tau$. We have
	\begin{align*}
	\frac12 \ge \tv(P,Q) &\ge \pr[X \leftarrow P]{\langle \nu , X \rangle \ge t} - \pr[X \leftarrow Q]{\langle \nu , X \rangle \ge t}\\
	&=  \pr[X \leftarrow P]{\langle \nu , X - \mu \rangle \ge t - \langle \nu, \mu \rangle} - \pr[X \leftarrow Q]{\langle \nu , X - \mu' \rangle \ge t - \langle \nu, \mu \rangle + \langle \nu, \mu-\mu' \rangle }\\
	(\text{set $t= \langle \nu, \mu \rangle - \frac{\tau}{2}$})~~&=  \pr[X \leftarrow P]{\langle \nu , X - \mu \rangle \ge -\frac{\tau}{2}} - \pr[X \leftarrow Q]{\langle \nu , X - \mu' \rangle \ge +\frac{\tau}{2} }\\
	&=  1-\pr[X \leftarrow P]{\langle \nu , X - \mu \rangle < -\frac{\tau}{2}} - \pr[X \leftarrow Q]{\langle \nu , X - \mu' \rangle \ge +\frac{\tau}{2} }\\
	(\text{Chebyshev's inequality})~~&\ge 1 - \frac{\ex[X \leftarrow P]{\langle \nu , X - \mu \rangle^2}}{(\tau/2)^2}  - \frac{\ex[X \leftarrow Q]{\langle \nu , X - \mu' \rangle^2}}{(\tau/2)^2} \\
	&= 1 - \frac{4}{\tau^2} \sum_{i=1}^d \ex[X \leftarrow P]{(X_i-\mu_i)^2}+ \ex[X \leftarrow Q]{(X_i-\mu'_i)^2}\\
	&\ge 1 - \frac{8d\sigma^2}{\tau^2}.
	\end{align*}
	Rearranging yields $\tau \le 4\sqrt{d\sigma^2}$, as required.
\end{proof}

\subsection{Gaussian Distributions}
\label{sec:gaussian-app}
We next give private algorithms for learning Gaussian distributions.
\begin{definition}
  A \emph{Gaussian distribution} $\mathcal{N}(\mu, \Sigma)$ in $\mathbb{R}^d$ is a distribution with PDF
  $$p(x) = \frac{\exp\left(-\frac{1}{2}(x - \mu)^T\Sigma^{-1}(x - \mu)\right)}{\sqrt{(2\pi)^d|\Sigma|}}.$$
\end{definition}

We describe covers for Gaussian distributions with known and unknown covariance.
\begin{lemma}
  \label{lem:gaussian-mean-cover}
  There exists an $\alpha$-cover of the set of Gaussian distributions $\mathcal{N}(\mu, I)$ in $d$ dimensions with $\|\mu\|_2 \leq R$ of size
  $$O\left(\frac{dR}{\alpha}\right)^d .$$
\end{lemma}
\begin{proof}
  It is well-known that estimating a Gaussian distribution with unknown mean in total variation distance corresponds to estimating $\mu$ in $\ell_2$-distance (see, e.g.,~\cite{DiakonikolasKKLMS16}).
  By the triangle inequality, in order to $\alpha$-cover the space, it suffices to $(\alpha/d)$-cover each standard basis direction.
  Since we know the mean in each direction is bounded by $R$, a simple additive grid in each direction with granularity $\Theta\left(\frac{\alpha}{d}\right)$ will suffice, resulting in a cover for each direction of size $O\left(\frac{dR}{\alpha}\right)$.
  Taking the Cartesian product over $d$ dimensions gives the desired result.
\end{proof}

\begin{lemma}\label{lem:gaussian-var-cover}
  There exists an $\alpha$-cover of the set of Gaussian distributions $\mathcal{N}(\mu, \Sigma)$ in $d$-dimensions with $\|\mu\|_2 \leq R$ and $I \preceq \Sigma \preceq \kappa I$ of size
  $$O\left(\frac{dR}{\alpha}\right)^d \cdot O\left(\frac{d\kappa}{\alpha}\right)^{d(d+1)/2}.$$
\end{lemma}
\begin{proof}
  The former term is obtained similarly to the expression in Lemma~\ref{lem:gaussian-mean-cover}.
  Since $I \preceq \Sigma$, we can still bound the total variation contribution by the $\ell_2$-distance between the mean vectors.
  We thus turn our attention to the latter term.
  To construct our cover, we must argue about the total variation distance between $\mathcal{N}(0, \Sigma)$ and $\mathcal{N}(0, \hat \Sigma)$.
  If $|\Sigma(i,j) - \hat \Sigma(i,j)| \leq \gamma$, and $I \preceq \Sigma$, Proposition 32 of~\cite{ValiantV10a} implies:
  $$\dtv(\mathcal{N}(0,\Sigma), \mathcal{N}(0, \hat \Sigma)) \leq O(d\gamma).$$
  We will thus perform a gridding, in order to approximate each entry of $\Sigma$ to an additive $O(\gamma) = O(\alpha/d)$.
  However, in order to ensure that the resulting matrix is PSD, we grid over entries of $\hat \Sigma$'s Cholesky decomposition, rather than grid for $\hat \Sigma$ itself.
  Since the largest element of $\Sigma$ is bounded by $\kappa$, the larest element of its Cholesky decomposition must be bounded by $\sqrt{\kappa}$.
  An additive grid over the range $[0,\sqrt{\kappa}]$ with granularity $O(\gamma/\sqrt{\kappa})$ suffices to get $\hat \Sigma$ which bounds the entrywise distance as $O(\gamma)$.
  This requires $O(d\kappa/\alpha)$ candidates per entry, and we take the Cartesian product over all $d(d+1)/2$ entries of the Cholesky decomposition, giving the desired result.
\end{proof}

In addition, we can obtain bounds of the VC dimension of the Scheff\'e sets of Gaussian distributions.

\begin{lemma}\label{lem:vc-lin-quad}
	The set of Gaussian distributions with fixed variance -- i.e., all $\mathcal{N}(\mu,I)$ with $\mu \in \mathbb{R}^d$ -- has VC dimension $d+1$.  Furthermore, the set of Gaussians with unknown variance -- i.e., all $\mathcal{N}(\mu,\Sigma)$ with $\mu \in \mathbb{R}^d$ and $\Sigma \in \mathbb{R}^{d \times d}$ positive definite -- has VC dimension $O(d^2)$.
\end{lemma}
\begin{proof}
	For Gaussians with fixed variance, the Scheff\'e sets correspond to linear threshold functions, which have VC dimension $d+1$. For Gaussians with unknown variance, the Scheff\'e sets correspond to quadratic threshold functions, which have VC dimension ${d+2 \choose 2} = O(d^2)$~\cite{Anthony95}.
\end{proof}

Combining the covers of Lemmas \ref{lem:gaussian-mean-cover}  and 
\ref{lem:gaussian-var-cover} and the VC bound of Lemma \ref{lem:vc-lin-quad} with Theorem \ref{thm:sel-inf} implies the following corollaries for Gaussian estimation.
\begin{corollary}
\label{cor:learn-gauss-mean}
  Suppose we are given a set of samples $X_1, \dots, X_n \sim P$, where $P$ is $\alpha$-close to a Gaussian distribution $\mathcal{N}(\mu, I)$ in $d$-dimensions with $\|\mu\| \leq R$.
  Then for any constant $\zeta > 0$, there exists an $\ve$-differentially private algorithm which outputs a Gaussian distribution $H^*$ such that $\dtv(P,H^*) \leq (6+2\zeta)\alpha$ with probability $\geq 9/10$, so long as
  $$n = \Omega\left(\frac{d}{\alpha^2} + \frac{d}{\alpha\varepsilon} \log\left(\frac{dR}{\alpha}\right)\right).$$
\end{corollary}

\begin{corollary}
\label{cor:learn-gauss-cov}
  Suppose we are given a set of samples $X_1, \dots, X_n \sim P$, where $P$ is $\alpha$-close to a Gaussian distribution $\mathcal{N}(\mu, \Sigma)$ in $d$-dimensions with $\|\mu\| \leq R$ and $I \preceq \Sigma \preceq \kappa I$.
  Then for any constant $\zeta > 0$, there exists an $\ve$-differentially private algorithm which outputs a Gaussian distribution $H^*$ such that $\dtv(P,H^*) \leq (6+2\zeta)\alpha$ with probability $\geq 9/10$, so long as
  $$n = \Omega\left(\frac{d^2}{\alpha^2} + \frac{1}{\alpha\varepsilon}\left(d \log\left(\frac{dR}{\alpha}\right) + d^2 \log\left(\frac{d\kappa}{\alpha}\right) \right)\right).$$
\end{corollary}

Similar to the product distribution case, these are the first $\tilde O(d)$ and $\tilde O(d^2)$ sample algorithms for learning Gaussians total variation distance under pure differential privacy, improving upon the concentrated differential privacy results of Kamath, Li, Singhal, and Ullman~\cite{KamathLSU19}.

\subsubsection{Gaussians with Unbounded Mean}
Extending Corollary \ref{cor:learn-gauss-mean}, we consider multivariate Gaussian hypotheses with known covariance and unknown mean, \emph{without} assuming bound on the mean (the parameter $R$ in the discussion above). To handle the unbounded mean we must relax to approximate differential privacy. 

In place of Lemma \ref{lem:gaussian-mean-cover}, we construct a locally small cover:
\begin{lemma}
	\label{lem:gaussian-mean-cover2}
	For any $d \in \mathbb{N}$ and $\alpha \in (0,1/30]$, there exists an $\alpha$-cover  $\mathcal{C}_\alpha$ of the set of Gaussian distributions $\mathcal{N}(\mu, I)$ in $d$ dimensions satisfying
	$$\forall \mu \in \mathbb{R}^d  ~~ \left|\left\{ H \in \mathcal{C}_\alpha : \tv(H,\mathcal{N}(\mu, I)) \le 7\alpha \right\}\right| \le 2^{15d}.$$
\end{lemma}
\begin{proof}
	For $\mu,\mu'\in\mathbb{R}^d$, we have 
	\begin{align*}
	\tv(\mathcal{N}(\mu,I),\mathcal{N}(\mu',I)) &= 2\pr{\mathcal{N}(0,1) \in \left[0,\frac12 \|\mu-\mu'\|_2\right]}\\
	&= \sqrt{\frac{2}{\pi}}\int_0^{\frac12 \|\mu-\mu'\|_2} e^{-x^2/2} \mathrm{d}x\\
	&\le \frac{\|\mu-\mu'\|_2}{\sqrt{2\pi}}.
	\end{align*}
	Furthermore, for any $c>0$, $$\tv(\mathcal{N}(\mu,I),\mathcal{N}(\mu',I)) \ge \left\{\begin{array}{cl} \frac{\|\mu-\mu'\|_2}{\sqrt{2\pi}} \cdot e^{-c^2/2} & \text{ if } \frac12\|\mu-\mu'\|_2 \le c \\ \frac{c \cdot e^{-c^2/2}}{\sqrt{2\pi}} & \text{ if } \frac12\|\mu-\mu'\|_2 \ge c \end{array}\right..$$
	
	Let $$\mathcal{C}_\alpha = \left\{\mathcal{N}\left(m \cdot \frac{\alpha\sqrt{8\pi}}{\sqrt{d}}, I\right) : m \in \mathbb{Z}^d\right\}.$$
	Fix $\mu \in \mathbb{R}^d$. Let $\mu^*= \mu \frac{\sqrt{d}}{\alpha \sqrt{8\pi}}\in\mathbb{R}^d$ and let $m \in \mathbb{Z}^d$ be $\mu^*$ rounded to the nearest integer coordinate-wise, so that $\|m-\mu^*\|_\infty \le \frac12$. Then \begin{align*}
		\tv\left(\mathcal{N}(\mu,I),\mathcal{N}\left(m \cdot \frac{\alpha\sqrt{8\pi}}{\sqrt{d}},I\right)\right) &= \tv\left(\mathcal{N}\left(\mu^* \cdot \frac{\alpha\sqrt{8\pi}}{\sqrt{d}},I\right),\mathcal{N}\left(m \cdot \frac{\alpha\sqrt{8\pi}}{\sqrt{d}},I\right)\right)\\
		&\le \frac{1}{\sqrt{2\pi}} \frac{\alpha\sqrt{8\pi}}{\sqrt{d}} \|\mu^*-m\|_2\\
		&\le \alpha,
	\end{align*}
	since $\|\mu^*-m\|_2 \le \sqrt{d}\|\mu^*-m\|_\infty \le \frac{\sqrt{d}}{2}$. This proves that $\mathcal{C}_\alpha$ is a $\alpha$-cover of $\{\mathcal{N}(\mu,I):\mu\in\mathbb{R}^d\}$.
	
	It remains to show that the cover is ``locally small''. Let $m'\in\mathbb{Z}^d$. Then
	\begin{align*}
		\tv\left(\mathcal{N}(\mu,I),\mathcal{N}\left(m' \cdot \frac{\alpha\sqrt{8\pi}}{\sqrt{d}},I\right)\right) &= \tv\left(\mathcal{N}\left(\mu^* \cdot \frac{\alpha\sqrt{8\pi}}{\sqrt{d}},I\right),\mathcal{N}\left(m' \cdot \frac{\alpha\sqrt{8\pi}}{\sqrt{d}},I\right)\right)\\
		&\ge  \frac{c \cdot e^{-c^2/2}}{\sqrt{2\pi}}  ~~~~\text{ if } \frac12\|\mu^*-m'\|_2 \frac{\alpha\sqrt{8\pi}}{\sqrt{d}} \ge c \\
		&> 7\alpha ~~~~\text{ if } \|\mu^*-m'\|_2 \ge 30 \frac{\sqrt{d}}{\sqrt{2\pi}},
	\end{align*}
	where the final inequality follows by setting $c=30\alpha \le 1$. Thus 
	\begin{align*}
	\left|\left\{ H \in \mathcal{C}_\alpha : \tv(H,\mathcal{N}(\mu, I)) \le 7\alpha \right\}\right| &\le \left|\left\{ m' \in \mathbb{Z}^d :  \|\mu^*-m'\|_2 < 30 \frac{\sqrt{d}}{\sqrt{2\pi}} \right\}\right|\\	&\le \left|\left\{ m' \in \mathbb{Z}^d :  \|m-m'\|_2 < 30 \frac{\sqrt{d}}{\sqrt{2\pi}}+\|\mu^*-m'\|_2 \right\}\right|\\
	&\le \left|\left\{ m' \in \mathbb{Z}^d :  \|m-m'\|_2 < 13\sqrt{d} \right\}\right|\\
	&\le \left|\left\{ w \in \mathbb{Z}^d :  \|w\|_1 < 13d \right\}\right|.
	\end{align*}
	Now we note that any $w \in \mathbb{Z}^d$ with  $\|w\|_1 \le r$ can be written as $w=x-y$ where $x,y \in \mathbb{Z}^d$ with $\sum_{i=1}^d x_i + y_i = r$ and, for all $i \in [d]$, we have $x_i \ge 0$ and $y_i \ge 0$. Instead of counting these $w$ vectors, we can count such $(x,y)$ vector pairs. We can interpret a pair of $x,y$ vectors as a way of putting $r$ balls into $2d$ bins or $r$ ``stars'' and $2d-1$ ``bars''. We can thus count $$\left|\left\{ w \in \mathbb{Z}^d :  \|w\|_1 < 13d \right\}\right| \le \left|\left\{ x,y \in \mathbb{Z}^d :  \|x\|_1+\|y\|_2 = 13d-1, x \ge 0, y \ge 0 \right\}\right| \le {15d-2 \choose 2d-1} \le 2^{15d}.$$
\end{proof}

Applying Theorem \ref{thm:sel-inf} with the cover of Lemma \ref{lem:gaussian-mean-cover2} and the VC bound from Lemma \ref{lem:vc-lin-quad} now yields an algorithm.

\begin{corollary}
	\label{cor:learn-gauss-mean2}
	Suppose we are given a set of samples $X_1, \dots, X_n \sim P$, where $P$ is a spherical Gaussian distribution $\mathcal{N}(\mu, I)$ in $d$-dimensions.
	Then there exists a $(\ve,\delta)$-differentially private algorithm which outputs a spherical Gaussian distribution $H^*$ such that $\dtv(P,H^*) \leq 7\alpha$ with probability $\geq 1-2^{-d}$, so long as
	$$n = \Omega \left(\frac{d}{\alpha^2} + \frac{d + \log(1/\delta)}{\alpha \ve}\right).$$
\end{corollary}

Karwa and Vadhan \cite{KarwaV18} give an algorithm for estimating a univariate Gaussian with unbounded mean. 
One can consider applying their algorithm independently to the $d$ coordinates (which is done in~\cite{KamathLSU19}), giving a sample complexity bound of $\tilde O\left(\frac{d}{\alpha^2} + \frac{d}{\alpha\varepsilon} + \frac{\sqrt{d}\log^{3/2}(1/\delta)}{\ve}\right)$, which our bound dominates except for very small values of $\alpha$.

\subsubsection{Univariate Gaussians with Unbounded Mean and Variance}

Our methods also allow us to derive learning algorithms for univariate Gaussians with unknown mean and variance.

\begin{lemma}\label{lem:univarGauss-cover2}
	For all $\alpha$ less then some absolute constant, there exists an $\alpha$-cover $\mathcal{C}_\alpha$ of the set of univariate Gaussian distributions satisfying $$\forall \mu,\sigma\in\mathbb{R} ~~~ \left| \left\{H \in \mathcal{C}_\alpha : \tv(H,\mathcal{N}(\mu,\sigma^2)) \le 7 \alpha \right\}\right|\le O(1).$$
\end{lemma}
\begin{proof}
	For all $\mu,\tilde\mu\in\mathbb{R}$ and all $\sigma,\tilde\sigma>0$, we have \cite[Thm 1.3]{DevroyeMR18b} $$\frac{1}{200} \min\left\{1, \max\left\{\frac{|\tilde\sigma^2-\sigma^2|}{\tilde\sigma^2},\frac{40|\tilde\mu-\mu|}{\tilde\sigma}\right\}\right\}\le \tv(\mathcal{N}(\mu,\sigma^2),\mathcal{N}(\tilde\mu,\tilde\sigma^2)) \le \frac{3|\tilde\sigma^2-\sigma^2|}{2\tilde\sigma^2} + \frac{|\tilde\mu-\mu|}{2\tilde\sigma}.$$
	Let $\beta=\alpha$ and $\gamma=\log(1+\alpha/2)$. Define the set of distributions $$\mathcal{C}_\alpha=\left\{ \mathcal{N}\left(\beta e^{\gamma n} m, e^{2\gamma n}\right) : n,m\in\mathbb{Z} \right\}.$$
	We first show that $\mathcal{C}_\alpha$ is an $\alpha$-cover: Let $\mu\in\mathbb{R}$ and $\sigma>0$. Let $n = \left[\frac{\log \sigma}{\gamma}\right]$ and $m = \left[\frac{\mu}{\beta e^{\gamma n}}\right]$, where $[x]$ denotes the nearest integer to $x$, satisfying $|x-[x]|\le\frac12$. Let $\tilde\sigma = e^{\gamma n}$ and $\tilde\mu= \beta e^{\gamma n} m$ so that $e^{-\gamma} \le \frac{\tilde\sigma^2}{\sigma^2}\le e^{\gamma}$ and $|\mu-\tilde\mu| \le \frac12 \beta e^{\gamma n} = \frac12 \beta \tilde\sigma$. Thus $\mathcal{N}(\tilde\mu,\tilde\sigma^2) \in \mathcal{C}_\alpha$ and $\tv(\mathcal{N}(\mu,\sigma^2),\mathcal{N}(\tilde\mu,\tilde\sigma^2)) \le \frac32 \left(e^\gamma - 1\right)+ \frac{\beta}{4} \le \alpha$, as required.
	
	It only remains to show that the cover size is locally small. Let $\mu\in\mathbb{R}$ and $\sigma>0$. 
	\begin{align*}
		\left| \left\{H \in \mathcal{C}_\alpha : \tv(H,\mathcal{N}(\mu,\sigma^2)) \le 7 \alpha \right\}\right| 
		&= \left| \left\{n,m \in \mathbb{Z} : \tv(\mathcal{N}\left(\beta e^{\gamma n} m, e^{2\gamma n}\right),\mathcal{N}(\mu,\sigma^2)) \le 7 \alpha \right\}\right|\\
		&\le \left| \left\{n,m \in \mathbb{Z} : \max\left\{\frac{|e^{2\gamma n}-\sigma^2|}{e^{2\gamma n}},\frac{40|\beta e^{\gamma n} m-\mu|}{e^{\gamma n}}\right\} \le 1400 \alpha \right\}\right|\\
		&=\left|\left\{n,m\in\mathbb{Z}: \begin{array}{c} \frac{ - \log(1 + 1400\alpha) }{2\gamma}  \le n - \frac{\log \sigma}{\gamma} \le \frac{ -\log(1 - 1400\alpha) }{2\gamma}  \\ -35 \frac{\alpha}{\beta} \le m- \frac{\mu}{\beta e^{\gamma n}} \le 35 \frac{\alpha}{\beta} \end{array}\right\}\right|\\
		&\le \left( \frac{ -\log(1 - 1400\alpha) }{2\gamma} - \frac{ -\log(1 + 1400\alpha) }{2\gamma} + 1\right) \cdot (35 - (-35)+1)\\
		&= \frac{1}{2\log(1+\alpha/2)} \log\left(\frac{1+1400\alpha}{1-1400\alpha}\right) \cdot 71 + 71\\
		&=O(1).
	\end{align*}
\end{proof}

Combining Lemma \ref{lem:univarGauss-cover2} with Lemma \ref{lem:vc-lin-quad} and Theorem \ref{thm:sel-inf} yields the following.

\begin{corollary}
  \label{cor:learn-gaussian-univ-infinite}
	Suppose we are given a set of samples $X_1, \dots, X_n \sim P$, where $P$ is a univariate Gaussian distribution $\mathcal{N}(\mu, \sigma^2)$.
	Then there exists a $(\ve,\delta)$-differentially private algorithm which outputs a univariate Gaussian distribution $H^*$ such that $\dtv(P,H^*) \leq 7\alpha$ with probability $\geq 9/10$, so long as
	$$n = \Omega \left(\frac{1}{\alpha^2} + \frac{\log(1/\delta)}{\alpha \ve}\right).$$
\end{corollary}
This sample complexity is comparable to to that of Karwa and Vadhan~\cite{KarwaV18}, who give an $(\ve,\delta)$-DP algorithm with sample complexity $\tilde O\left(\frac{1}{\alpha^2} + \frac{1}{\alpha\ve} + \frac{\log (1/\delta)}{\ve}\right)$.

\subsection{Sums of Independent Random Variables}
\label{sec:sirv-app}
In this section, we apply our results to distribution classes which are defined as the sum of independent (but not necessarily identical) distributions.
These are all generalizations of the classical Binomial distribution, and they have enjoyed a great deal of study into the construction of sparse covers.
To the best of our knowledge, we are the first to provide private learning algorithms for these classes.

We start with the Poisson Binomial distribution.
\begin{definition}
  A \emph{$k$-Poisson Binomial Distribution} ($k$-PBD) is the sum of $k$ independent Bernoulli random variables.
\end{definition}

We next consider sums of independent integer random variables, which generalize PBDs (which correspond to the case $d=2$).
\begin{definition}
  A \emph{$(k,d)$-Sum of Independent Integer Random Variables} ($(k,d)$-SIIRV) is the sum of $k$ independent random variables over $\{0, \dots, d-1\}$.
\end{definition}

Finally, we consider Poisson Multinomial distributions, which again generalize PBDs (which, again, correspond to the case $d=2$).
\begin{definition}
  A \emph{$(k,d)$-Poisson Multinomial Distribution} ($(k,d)$-PMD) is the sum of $k$ independent $d$-dimensional categorical random variables, i.e., distributions over $\{e_1, \dots, e_d\}$, where $e_i$ is the $i$th basis vector.
\end{definition}

We start with a covering result for SIIRVs (including the special case of PBDs), which appears in~\cite{DiakonikolasKS16a}.
Previous covers for PBDs and SIIRVs appear in~\cite{DaskalakisP09, DaskalakisP15a,DaskalakisDOST13}.
\begin{lemma}[\cite{DiakonikolasKS16a}]
  There exists an $\alpha$-cover of the set of $(k,d)$-SIIRVs of size
  $$k \cdot 2^{O(d \log^2 (1/\alpha)+ d \log^2 d)} .$$
\end{lemma}

Using this cover, we can apply Corollary~\ref{cor:covertolearn} to attain the following learning result for PBDs and SIIRVs.
\begin{corollary}
  \label{cor:learn-siirv}
  Suppose we are given a set of samples $X_1, \dots, X_n \sim P$, where $P$ is $\alpha$-close to a $(k,d)$-SIIRV.
  Then for any constant $\zeta > 0$, there exists an $\ve$-differentially private algorithm which outputs a $(k,d)$-SIIRV $H^*$ such that $\dtv(P,H^*) \leq (6+2\zeta)\alpha$ with probability $\geq 9/10$, so long as
  $$n = \Omega\left(\left(\log k + d\log^2(1/\alpha) + d \log^2 d\right)\left(\frac{1}{\alpha^2} + \frac{1}{\alpha \ve}\right)\right).$$
\end{corollary}

Next, we move on to PMDs.
The following cover does not appear verbatim in any single location, but is a combination of results from a few different sources.
The proofs for the best bounds on first term appears in~\cite{DaskalakisDKT16}, the second in~\cite{DaskalakisKT15}, and the third in~\cite{DiakonikolasKS16c}.
Larger covers previously appeared in~\cite{DaskalakisP08, DaskalakisP15b}.
\begin{lemma}[\cite{DaskalakisKT15, DaskalakisDKT16,DiakonikolasKS16c}]
  For any $d > 2$, there exists an $\alpha$-cover of the set of $(k,d)$-PMDs of size 
  $$k^{O(d)} \cdot \min\left\{2^{\poly(d/\alpha)}, (1/\alpha)^{O(d\log(d/\alpha)/\log \log (d/\alpha))^{d-1}}\right\}.$$
\end{lemma}

This implies the following learning result for PMDs.
\begin{corollary}
  \label{cor:learn-pmd}
  Suppose we are given a set of samples $X_1, \dots, X_n \sim P$, where $P$ is $\alpha$-close to a $(k,d)$-PMD, for any $d > 2$.
  Then there exists an $\ve$-differentially private algorithm which outputs a $(k,d)$-PMD $H^*$ such that $\dtv(P,H^*) \leq (6+2\zeta)\alpha$ with probability $\geq 9/10$, so long as
  $$n = \tilde \Omega\left(\left(d \log k + \min\left\{\poly\left(\frac{d}{\alpha}\right),O\left(\frac{d \log(d/\alpha)}{\log \log (d/\alpha)}\right)^{d-1}\cdot \log(1/\alpha)\right\}\right)\left(\frac{1}{\alpha^2} + \frac{1}{\alpha\ve}\right)\right). $$
\end{corollary}

\subsection{Piecewise Polynomials}
\label{sec:pp-app}
In this section, we apply our results to semi-agnostically learn piecewise polynomials.
This class of distributions is very expressive, allowing us to approximate a wide range of natural distribution classes.

\begin{definition}
  A \emph{$(t,d,k)$-piecewise polynomial distribution} is a distribution $P$ over $[k]$, such that there exists a partition of $[k]$ into $t$ disjoint intervals $I_1, \dots, I_t$ such that on each interval $I_j \subseteq [k]$, the probability mass function of $P$ takes the form $p_j(x) = \sum_{i=0}^d c_i^{(j)} x^i$ for some coefficients $c_i^{(j)}$, for all $x \in I_j$.
\end{definition}

We construct a cover for piecewise polynomials.
\begin{lemma}
  There exists a universal constant $c > 0$ such that there is an $\alpha$-cover of the set of $(t,d,k)$-piecewise polynomials of size
  $$\binom{k}{t - 1} \cdot \left( \frac{tk\cdot e^{cd^{1/2}}}{\alpha}\right)^{(d+1)t}.$$
\end{lemma}
\begin{proof}
We specify an element of the cover by 
\begin{enumerate}
\item Selecting one of $\binom{k}{t-1}$ partitions of $[k]$ into $t$ intervals $I_1, \dots, I_t$, and
\item For each interval $I_j$, selecting an element of an $(\alpha / t)$-cover $\mathcal{C}_j$ of the set of degree-$d$ polynomials over $I_j$ which are uniformly bounded by $1$.
\end{enumerate}
The total size of the cover is $\binom{k}{t-1} \prod_{j = 1}^t |\mathcal{C}_j|$. The theorem follows from Proposition~\ref{prp:one-poly-cover} below, which constructs an $(\alpha/t)$-cover $\mathcal{C}_j$ of size at most $\left( \frac{tk\cdot e^{cd^{1/2}}}{\alpha}\right)^{d+1}$ for every interval $I_j$.

\end{proof}

\begin{proposition} \label{prp:one-poly-cover}
There exist constants $b, c > 0$ for which the following holds. Let $I \subseteq [k]$ be an interval and let $\mathcal{P}$ be the set of polynomials $p : I \to \mathbb{R}$ of degree $d$ such that $|p(x)| \le 1$ for all $x \in I$. There exists an $\alpha$-cover of $\mathcal{P}$ of size
\[\min \left\{ \left(\frac{2k}{\alpha}\right)^{|I|} , \left(\frac{ckd^2 \cdot e^{bd^2/|I|}}{\alpha}\right)^{d+1}\right\}.\]
\end{proposition}

The proof of Proposition~\ref{prp:one-poly-cover} relies on two major results in approximation theory, which we now state.

\begin{lemma}[Duffin and Schaeffer~\cite{DuffinS41}] \label{lem:duffin-schaeffer}
Let $p :[-1, 1] \to \mathbb{R}$ be a polynomial such that $|p(x)| \le 1$ for all $x$ of the form $x = \cos(j \pi / d)$ for $j = 0, 1, \dots, d$. Then $|p'(x)| \le d^2$ for all $x \in [-1, 1]$.
\end{lemma}

\begin{lemma}[Coppersmith and Rivlin~\cite{CoppersmithR92}] \label{lem:coppersmith-rivlin}
There exist constants $a, b > 0$ for which the following holds. Let $p : \mathbb{R} \to \mathbb{R}$ be a polynomial of degree $d$, and suppose that $|p(t)| \le 1$ for all $t = 0, 1, \dots, m$. Then $|p(t)| \le a\exp(bd^2 / m)$ for all $t \in [0, m]$.
\end{lemma}

\begin{proof}[Proof of Proposition~\ref{prp:one-poly-cover}]
We consider two cases, corresponding to the two terms in the minimum. First, consider the function $f : I \to \mathbb{R}$ where $f(t)$ is obtained by rounding $p(t)$ to the nearest multiple of $\alpha / k$. Then $f$ satisfies $\sum_{t \in I} |f(t) - p(t)| \le \alpha$. There are at most $(2k/\alpha)^{|I|}$ functions $f$ which can be constructed this way, giving the first term in the maximum.

For the second term, we construct a cover for $\mathcal{P}$ by approximately interpolating through $d+1$ carefully chosen points in the \emph{continuous} interval corresponding to $I$. By applying an affine shift, we may assume that $I = \{0, 1, \dots, m\}$ for some integer $m \le k-1$. Let $p \in \mathcal{P}$ and for $x \in [0, m]$ let $\hat{p}(x)$ be the value of $p(x)$ rounded to the nearest integer multiple of $\alpha / (2kd^2)$. Let $q : [0, m] \to \mathbb{R}$ be the unique degree-$d$ polynomial obtained by interpolating through the points $(x_j, \hat{p}(x_j))$ where $x_j = m(1 + \cos(j \pi / d)) / 2$ for $j = 0, 1, \dots, d$.

We first argue that the polynomial $q$ so defined satisfies $\sum_{t = 0}^m |p(t) - q(t)| \le \alpha$. Let $r(x) = p(x) - q(x)$ for $x \in [0, m]$. Then by construction, $|r(x_j)| \le \alpha / (2kd^2)$ for all interpolation points $x_j$. By the Duffin-Schaeffer Inequality (Lemma~\ref{lem:duffin-schaeffer}), we therefore have $|r'(x)| \le \frac{\alpha}{km}$ for all $x \in [0, m]$. By the Fundamental Theorem of Calculus, $r(t) = r(0) + \int_{0}^t r'(t) \ dt$ satisfies $|r(t)| \le (t+1) \cdot \frac{\alpha}{km} \le \alpha / k$, and hence $\sum_{t = 0}^m |r(t)| \le \alpha$.

We now argue that the set of polynomials $q$ that can be constructed in this fashion has size $(ckd^2 \exp(bd^2/m) / \alpha)^{d+1}$. By the Coppersmith-Rivlin Inequality (Lemma~\ref{lem:coppersmith-rivlin}), there are constants $a, b > 0$ such that $|p(x)| \le a\exp(bd^2/m)$ for all $x \in [0, m]$. Therefore, for each $p \in \mathcal{P}$ and each interpolation point $x_j$, there are at most $4a\cdot kd^2\exp(bd^2/m) / \alpha$ possible values that $\hat{p}(x_j)$ can take. Hence, the polynomial $q$ can take one of at most $(4a\cdot kd^2\exp(bd^2/m) / \alpha)^{d+1}$ possible values, as we wanted to show.
\end{proof}

\begin{lemma}\label{lem:pp-vc}
	The VC dimension of $(t,d,k)$-piecewise polynomial distributions is at most $2t(d+1)$.
\end{lemma}
\begin{proof}
	Consider two piecewise polynomial distributions. The difference between their probability mass functions is a piecewise polynomial of degree $\le d$. The number of intervals needed to represent this piecewise function is $\le 2t$. It follows that this difference can change sign at most $2td+ 2t-1$ times -- each polynomial can change sign at most $d$ times and the sign can change at the interval boundaries. Thus such a function cannot label $2td+ 2t + 1$ points with alternating signs, which implies the VC bound.
\end{proof}

As a corollary, we obtain the following learning algorithm.
\begin{corollary}
\label{cor:learn-pp}
  Suppose we are given a set of samples $X_1, \dots, X_n \sim P$, where $P$ is $\alpha$-close to a $(t,d,k)$-piecewise polynomial.
  Then there exists an $\ve$-differentially private algorithm which outputs a $(t,d,k)$-piecewise polynomial $H^*$ such that $\dtv(P,H^*) \leq (6+2\zeta)\alpha$ with probability $\geq 9/10$, so long as
  $$n = \Omega\left( \frac{(d+1)t}{\alpha^2} 
  + \frac{(d+1)t}{\alpha\varepsilon} \cdot \left(\sqrt{d+1} \log k + \log\left(\frac{t}{\alpha}\right)\right)\right).$$
\end{corollary}

We compare with the work of Diakonikolas, Hardt, and Schmidt~\cite{DiakonikolasHS15}.
They present an efficient algorithm for $(t,1,k)$-piecewise polynomials, with sample complexity $\tilde O\left(\frac{t}{\alpha^2} + \frac{t\log k}{\alpha \eps}\right)$, which our algorithm matches\footnote{As stated in~\cite{DiakonikolasHS15}, their algorithm guarantees approximate differential privacy, but swapping in an appropriate pure DP subroutine gives this result.}.
They also claim their results extend to $(t,d,k)$-piecewise polynomials, though no theorem statement is provided.
While we have not investigated the details of this extension, we believe the resulting sample complexity should be qualitatively similar to ours, plausibly with the factor of $\frac{t(d+1)^{3/2}\log k}{\alpha \ve}$ replaced by $\frac{t(d+1)\log k}{\alpha \ve}$.

\subsection{Mixtures}
\label{sec:mix-app}
In this section, we show that our results immediately extend to learning mixtures of classes of distributions.

\begin{definition}
  Let $\mathcal{H}$ be some set of distributions.
  A \emph{$k$-mixture of $\mathcal{H}$} is a distribution with density $\sum_{i=1}^k w_i P_i$, where each $P_i \in \mathcal{H}$.
\end{definition}
Our results follow roughly due to the fact that a cover for $k$-mixtures of a class can be written as the Cartesian product of $k$ covers for the class.
More precisely, we state the following result which bounds the size of the cover of the set of $k$-mixtures.

\begin{lemma}
  Consider the class of $k$-mixtures of $\mathcal{H}$, where $\mathcal{H}$ is some set of distributions.
  There exists a $2\alpha$-cover of this class of size $\left|\mathcal{C}_\alpha\right|^k \left(\frac{k}{2\alpha} + 1\right)^{k-1}$, where $\mathcal{C}_\alpha$ is an $\alpha$-cover of $\mathcal{H}$.
\end{lemma}
\begin{proof}
  Each element in the cover of the class of mixtures will be obtained by taking $k$ distributions from $\mathcal{C}_\alpha$, in combination with $k$ mixing weights, which are selected from the set $\left\{0, \frac{2\alpha}{k}, \frac{4\alpha}{k}, \dots, 1\right\}$, such that the sum of the mixing weights is $1$.
  The size of this cover is $|\mathcal{C}_\alpha|^k \cdot \left(\frac{k}{2\alpha} + 1\right)^{k-1}$.
  We reason about the accuracy of the cover as follows.
  Fix some mixture of $k$ distributions as $\sum_{i=1}^k w^{(1)}_i P^{(1)}_i$, and we will reason about the closest element in our cover, $\sum_{i=1}^k w^{(2)}_i P^{(2)}_i$.
  By triangle inequality, we have that
  $$\dtv\left(\sum_{i=1}^k w^{(1)}_i P^{(1)}_i, \sum_{i=1}^k w^{(2)}_i P^{(2)}_i\right) \leq \sum_{i=1}^k \frac{1}{2}\left|w^{(1)}_i - w^{(2)}_i\right| +  w^{(1)}_i \dtv\left(P^{(1)}_i, P^{(2)}_i\right). $$
  Since $\mathcal{C}_\alpha$ is an $\alpha$-cover and $\sum_{i=1}^k w^{(1)}_i = 1$, the total variation distance incurred by the second term will be at most $\alpha$.
  As for the mixing weights, note that for the first $k-1$ weights, the nearest weight is at distance at most $\frac{\alpha}{k}$, contributing a total of less than $\frac{\alpha}{2}$.
  The last mixing weight can be rewritten in terms of the sum of the errors of the other mixing weights, similarly contributing another total of less than $\frac{\alpha}{2}$.
  This results in the total error being at most $2\alpha$, as desired.
\end{proof}

With this in hand, the following corollary is almost immediate from Corollary~\ref{cor:covertolearn}.
The factor of $(9+3\zeta)\alpha$ (as opposed to $(6+2\zeta)\alpha$) is because the closest distribution in the cover of mixture distributions is $3\alpha$-close to be $P$ (rather than $2\alpha$).

\begin{corollary}
  \label{cor:learn-mixture}
  Let $X_1, \dots, X_n \sim P$, where $P$ is $\alpha$-close to a $k$-mixture of distributions from some set $\mathcal{H}$.
  Let $\mathcal{C}_\alpha$ be an $\alpha$-cover of the set $\mathcal{H}$, and $\zeta > 0$ be a constant.
  There exists an $\ve$-differentially private algorithm which outputs a distribution which is $(9+3\zeta)\alpha$-close to $P$ with probability $\geq 9/10$, as long as
  $$n = \Omega\left(\left(k \log |\mathcal{C}_\alpha| + k \log (k/\alpha)\right)\left(\frac{1}{\alpha^2} + \frac{1}{\alpha\ve}\right)\right).$$
\end{corollary}

For example, instantiating this for mixtures of Gaussians (and disregarding terms which depend on $R$ and $\kappa$), we get an algorithm with sample complexity $\tilde O\left(\frac{kd^2}{\alpha^2} + \frac{kd^2}{\alpha\ve}\right)$.

\subsection{Supervised Learning}
\label{sec:pac-app}

We describe an application of our results to the task of binary classification, as modeled by differentially private PAC learning~\cite{KasiviswanathanLNRS11}. Let $\mathcal{F} = \{f : X \to \{0, 1\}\}$ be a publicly known \emph{concept class} of Boolean functions over a domain $X$. Let $P$ be an unknown probability distribution over $X$, and let $f$ be an unknown function from $\mathcal{F}$. Given a sequence $\{(x_i, f(x_i))\}_{i = 1}^n$ of i.i.d.\ samples from $P$ together with their labels under $f$, the goal of a PAC learner $L$ is to identify a hypothesis $h : X \to \{0, 1\}$ such that $\Pr_{x \sim P} [h(x) \ne f(x)] \le \alpha$ for some error parameter $\alpha > 0$. We say that $L$ is \emph{$(\alpha, \beta)$-accurate} if for every $f \in \mathcal{F}$ and every distribution $P$, it is able to identify such a hypothesis $h$ with probability at least $1 - \beta$ over the choice of the sample and any internal randomness of $L$.

One of the core results of statistical learning theory is that the sample complexity of \emph{non-private} PAC learning is characterized, up to constant factors, by the VC dimension of the concept class $\mathcal{F}$. When one additionally requires the learner $L$ to be differentially private with respect to its input sample, such a characterization is unknown. However, it is known that the sample complexity of private learning can be arbitrarily higher than that of non-private learning. For example, when $\mathcal{F} = \{f_t : t \in X\}$ is the class of threshold functions defined by $f_t(x) = 1 \iff x \le t$ over a totally ordered domain $X$, the sample complexity of PAC learning under the most permissive notion of $(\eps, \delta)$-differential privacy is $\Omega(\log^* |X|)$~\cite{BunNSV15, AlonLMM19}. Meanwhile, the VC dimension of this class, and hence the sample complexity of non-private learning, is a constant independent of $|X|$.

While this separation shows that there can be a sample cost of privacy for PAC learning, this cost can be completely eliminated if the distribution $P$ on examples is known. This was observed by Beimel, Nissim, and Stemmer~\cite{BeimelNS16}, who showed that if a good approximation to $P$ is known, e.g., from public unlabeled examples or from differentially private processing of unlabeled examples, then the number of labeled examples needed for private PAC learning is only $O(VC(\mathcal{F}))$. 

\begin{theorem} \label{thm:public}
Let $\eps> 0$, $\mathcal{F} = \{f : X \to \{0, 1\}\}$, and $P$ be a publicly known distribution over $X$. For $n = O\left(\frac{1}{\alpha^2\eps}(VC(\mathcal{F}) \log(1/\alpha) + \log(1/\beta))\right)$, there exists an $\eps$-differentially private algorithm $L : (X \times \{0, 1\})^n \to \mathcal{F}$ such that for every $f \in \mathcal{F}$, with probability at least $1-\beta$ over the choice of $x_1, \dots, x_n \gets P$, we have that $L((x_1, f(x_1)), \dots, (x_n, f(x_n)))$ produces $h \in \mathcal{F}$ such that $\Pr_{x \sim P} [f(x) \ne h(x)] \le \alpha$.
\end{theorem}

Our results suggest a natural two-step algorithm for private PAC learning when the distribution $P$ itself is not known, but is known to (approximately) come from a set of distributions $\mathcal{H}$: The algorithm first uses private hypothesis selection to select $\hat{H}$ with $\dtv(P, \hat{H}) \le \alpha / 2$, and then runs the algorithm of~\cite{BeimelNS16} using $\hat{H}$ in place of $P$ with error parameter $\alpha / 2$. Using the fact that $\dtv(P, \hat{H}) \le \alpha / 2$ implies $|\Pr_{x \sim P}[f(x) \ne h(x)] - \Pr_{x \sim \hat{H}}[f(x) \ne h(x)]| \le \alpha / 2$, the following result holds by combining Theorem~\ref{thm:public} with Corollary~\ref{cor:covertolearn}.

\begin{corollary}
  \label{cor:pac}
  Let $\mathcal{H}$ be a set of distributions over $X$ with an $\alpha$-cover $\mathcal{C}_\alpha$. Let $P$ be a distribution over $X$ with $\dtv(P, \mathcal{H}) \le \alpha/(4(3+\zeta))$. Then for
\[n = O \left( \frac{\log |\mathcal{C}_\alpha|}{\alpha^2} + \frac{\log |\mathcal{C}_\alpha|}{\alpha \varepsilon } + \frac{VC(\mathcal{F})\log(1/\alpha)}{\alpha^2\eps} \right)\]
there exists an $\eps$-differentially private algorithm $L : (X \times \{0, 1\})^n \to \mathcal{F}$ such that for every $f \in \mathcal{F}$, with probability at least $3/4$ over the choice of $x_1, \dots, x_n \gets P$, we have that $L((x_1, f(x_1)), \dots, (x_n, f(x_n)))$ produces $h \in \mathcal{F}$ such that $\Pr_{x \sim P} [f(x) \ne h(x)] \le \alpha$.
\end{corollary}

Theorem~\ref{thm:public} can, of course, also be combined with the more refined guarantees of Theorem~\ref{thm:sel-inf}. As an example application, combining Theorem~\ref{thm:public} with Corollary~\ref{cor:learn-gauss-mean2} gives a $(\ve, \delta)$-differentially private algorithm for learning one-dimension thresholds with respect to univariate Gaussian distributions on the reals. In contrast, this task is impossible without making distributional assumptions.

\section{Conclusions}
\label{sec:conclusions}
In this paper, we presented differentially private methods for hypothesis selection.
The sample complexity can be bounded by the logarithm of the number of hypotheses.
This allows us to provide bounds on the sample complexity of (semi-agnostically) learning a class which depend on the logarithm of the covering number, complementing known lower bounds which depend on the logarithm of the packing number.
There are many interesting questions left open by our work, a few of which we outline below.
\begin{enumerate}
  \item Our algorithms for learning classes of distributions all use cover-based arguments, and thus are not computationally efficient.
    For instance, we provide the first $\tilde O(d)$ sample complexity upper bound on $\varepsilon$-differentially privately learning a product distribution and Gaussian with known covariance.
    One interesting question is whether there is an efficient algorithm which achieves this sample complexity.
  \item The running time of our method is quadratic in the number of hypotheses -- is it possible to reduce this to a near-linear time complexity?
  \item Our main theorem obtains an approximation factor which is arbitrarily close to $3$, which is optimal for this problem, even without privacy. 
    This factor can be reduced to 2 if one is OK with outputting a mixture of hypotheses from the set~\cite{BousquetKM19}.
    Is this achievable with privacy constraints?
\end{enumerate}

\if\nips0
\section*{Acknowledgments}
The authors would like to thank Shay Moran for bringing to their attention the application to PAC learning mentioned in Section~\ref{sec:pac-app}, Jonathan Ullman for asking questions which motivated Remark~\ref{rem:separation}, and Cl\'ement Canonne for assistance in reducing the constant factor in the approximation guarantee.
\fi

\bibliographystyle{alpha}
\bibliography{biblio}

\end{document}